\documentclass[a4paper]{amsart}    

\usepackage{graphicx}          

\usepackage{mathrsfs}
\usepackage{amsmath} 
\usepackage{amsfonts} 
\usepackage{amssymb} 
\usepackage{mathtools} 

\usepackage{xcolor}

\usepackage{amsthm}
\usepackage{thmtools} 

\theoremstyle{plain}
\newtheorem{exmp}{Example}[section]

\newtheorem{rem}{Remark}
\newtheorem{thm}{Theorem}[section]
\newtheorem{lem}[thm]{Lemma}
\newtheorem{prop}[thm]{Proposition}
\newtheorem{cor}{Corollary}
\newtheorem{hypo}{Assumption}

\theoremstyle{definition}
\newtheorem*{defn}{Definition}

\newcommand{\I}{\mathrm{\mathbf{I}}}
\newcommand{\0}{\mathbf{0}}

\newcommand{\card}[1]{|#1|}

\DeclareMathOperator*{\argmin}{arg\,min}

\renewcommand{\ker}[1]{\mathrm{Ker}\left\{#1\right\}}
\newcommand{\im}[1]{\mathrm{Im}\left\{#1\right\}}

\newcommand{\real}{\mathbb{R}}
\newcommand{\smallmat}[1]{\left[ \begin{smallmatrix}#1\end{smallmatrix} \right]} 

\usepackage{stackrel}
\usepackage{enumerate}

\renewcommand\H{\mathrm{H}}
\newcommand\N{\0} 

\newcommand{\Span}[1]{\mathrm{span}\left(#1\right)}

\newcommand{\fs}[1]{\mathrm{\bf #1}} 
\newcommand{\es}[1]{\mathcal{#1}} 

\newcommand{\locglobUnc}{Prop.~2.1}
\newcommand{\characIR}{Th.~3.2}
\newcommand{\matview}{Subsec.~4.1}
\newcommand{\secdeg}{Subsec.~4.2}

\newcommand{\sing}{IR}
\newcommand{\Sing}{IR}
\newcommand{\singy}{IR}
\newcommand{\Singy}{IR}

\newif\ifext
\exttrue

\begin{document}
	\makeatletter
	\let\@mkboth\@gobbletwo
	\let\@oddhead\thepage
	\let\@evenhead\thepage
	\makeatother


    \title{Input Redundancy under Input and State Constraints \ifext{\newline(Extended version of \cite{TREGOUET2024})}\fi}


    \author{Jean-François Trégouët}
    \address{Univ Lyon, INSA Lyon, Universit\'{e} Claude Bernard Lyon 1, Ecole Centrale de Lyon, CNRS, Amp\`{e}re UMR5005, F-69621 Villeurbanne, France}
    \email{jean-francois.tregouet@insa-lyon.fr}
 
    \author{Jérémie Kreiss}
 	\address{Universit\'e de Lorraine, CNRS, CRAN, F-54000 Nancy, France.}
 	\email{jeremie.kreiss@univ-lorraine.fr}
 	\date{\today}
    




\maketitle

	{\tt
	
	This document is an extended version of the following paper accepted for publication in Automatica:
	
	\begin{center}\it
		
		Jean-François Trégouët and Jérémie Kreiss. Input redundancy under input and state constraints. \begin{it}
				Automatica
			\end{it}, vol. 159, p. 111344, 2024, ISSN 0005-1098, https://doi.org/10.1016/j.automatica.2023.111344.

	\end{center}
}

\textcolor{blue}{Additional material are printed in blue}.

   \begin{abstract}
	For a given unconstrained dynamical system, input redundancy has been recently redefined as the existence of distinct inputs producing identical output for the same initial state. By directly referring to signals, this definition readily applies to any input-to-output mapping. As an illustration of this potentiality, this paper tackles the case where input and state constraints are imposed on the system. This context is indeed of foremost importance since input redundancy has been historically regarded as a way to deal with input saturations.
	An example illustrating how constraints can challenge redundancy is offered right at the outset. A more complex phenomenology is highlighted. This motivates the enrichment of the existing framework on redundancy. Then, a sufficient condition for redundancy to be preserved when imposing constraints is offered in the most general context of arbitrary constraints. It is shown that redundancy can be destroyed only when input and state trajectories lie on the border of the set of constraints almost all the time. Finally, those results are specialized and expanded under the assumption that input and state constraints are linear.
\end{abstract}

\section{Introduction}

This paper is a follow-up of \cite{Kreiss2020}, where new definitions, taxonomy and characterizations are proposed around the notion of input redundancy (IR).\footnote{Throughout this paper, IR stands either for ``input redundancy'' or for ``input redundant'' depending on the context.} This paper expands on this topic, by tackling the case where the system is affected by input and state constraints. 

Many systems are equipped by more actuators than strictly needed to meet the control objectives. Such a design has many technological advantages: Examples include state-of-health and/or thermal management, resilience to failure, enhanced control capabilities, see \cite{Tregoueet2019,Bouarfa2017,Tregouet2015,Johansen2013,Huang2007} to cite a few. Those systems are over-actuated. They belong to the class of IR systems.

Since the early nineties, literature devoted to IR systems have kept growing, see \cite{Durham1993,Enns1998,Harkegard2005,Zaccarian2009,Johansen2013,Cocetti2018} and references therein. Many control designs have been proposed in various fields, including aerospace and aeronautics \cite{Tregouet2015,Oppenheimer2010}, marine vessels \cite{Fossen2006} and, more recently, power electronics \cite{Kreiss2019,Tregoueet2019,Bouarfa2017,Kreiss2021}. In addition to that, significant contributions have been obtained on the methodological side \cite{Galeani2015,Serrani2012,Zaccarian2009,Johansen2004,Kreiss2020}. The goal is to enrich control theory by generalizing and formalizing ad hoc achievements.

Along this last line, the efforts toward a clear definition of IR, together with tractable characterizations, are of foremost importance. This topic is indeed crucial, as it impacts not only how dynamical systems are classified but also how the related control problem is tackled \cite{Kreiss2020}.
This paper contributes in this direction. It is concerned with the linear, proper and time-invariant dynamical system $\Sigma$ governed by the following equations
\begin{subequations} \label{eq:sys}
    \begin{align}\label{eq:sysAB}
        \dot{x}(t) & = Ax(t)+Bu(t), \; x(0)\eqqcolon x_0, \\ \label{eq:sysCD}
        y(t)       & = Cx(t)+Du(t),
    \end{align}
\end{subequations}
for some quadruple $(A,B,C,D)$ of appropriate dimensions. Here, vectors $u(t)\in\real^m$, $x(t)\in\real^n$ and $y(t)\in\real^p$ are the input, the state and the output at time $t$, respectively.

For the time being, let us assume that $\Sigma$ is not affected by any constraints. For this class of system, there exist various definitions of IR. A brief overview of them is now offered, starting from the observation that the literature is divided into two branches. Each one has its own definition of IR.
\begin{enumerate}[(i)]
    \item \emph{Input-to-state methods} focus on \eqref{eq:sysAB} and apply for non injective matrix $B$. The first step is to decompose $B u$ as $B_\tau \tau$ with $\tau = B_u u$ where $B_\tau$ is injective, i.e. $B$ is factorized as $B_\tau B_u$ \cite{Harkegard2005}. Signal $\tau$ usually refers to the overall contribution of the inputs. Then, a two-layers control scheme is designed: The high level controller delivering $\tau$ feeds the low-level allocator which select the optimal input $u$ among the ones that satisfy $\tau = B_u u$. The allocator is as follows:
          \begin{equation}\label{eq:optim}
              \tau \mapsto u(\tau)\in\argmin_{\hat u} J(\hat u) \;\text{s.t.}\; B_u \hat u = \tau,
          \end{equation}
          where $J$ is a real-valued cost function. Usually, no closed-form expression of $\tau\mapsto u(\tau)$ exists, so that computation of $u$ is performed online, either by solving \eqref{eq:optim} parametrized by $\tau$ at each instant time \cite{Johansen2013}, or by implementing a controller converging to the optimal solution \cite{Johansen2004}.
          In this framework, a system is IR if this strategy can be implemented, that is, if $B$ is non-injective.
    \item \emph{Input-to-output methods} enlarge the scope of the analysis by dealing with the input-to-output mapping, i.e. \eqref{eq:sysCD} is now taken into account \cite{Galeani2015,Serrani2012,Zaccarian2009,Kreiss2020}. From the control point of view, an important achievement of this line of research is \cite{Galeani2015}, where the output regulation problem is revisited by adding an optimizer to the classical control structure. This optimizer selects the best steady-state among those achieving exact output regulation. Various definitions of IR underpin those contributions. Ultimately, it is proposed in \cite{Kreiss2020} to define IR as the negation of left-invertibility, i.e. IR means that the input-to-output mapping is \emph{not} injective for some $x_0$.\footnote{Interested reader is referred to \cite[Sec.~5]{Kreiss2020} for a comprehensive comparison of the main definitions of IR coexisting in this line of research.}
\end{enumerate}
To sum up, IR is defined as non injectivity of either the input-to-state or the input-to-output mappings. For strictly proper system, note that the latter is more general than the former: If the input-to-state mapping is non injective, then the input-to-output mapping also enjoys the same property, a fortiori. To put it in another way, the utmost merit of the input-to-output methods is to consider the case of distinct inputs leading to the same output for the same initial state, even if the induced state trajectories are distinct.

So far, it is assumed that $\Sigma$ is free from input or state constraints. However, from the very beginning, IR has been considered as a way to deal with input constraints by redistributing overall control effort among actuators to avoid saturation. Earliest input-to-state methods add the following condition to the set of constraints of \eqref{eq:optim}, see \cite{Durham1993,Bordignon1995,Enns1998} and the survey \cite{Johansen2013}:
\begin{subequations}\label{eq:impstatconst}
    \begin{equation} \label{eq:impconst}
        u(t)\in\es{U},
    \end{equation}
    where $\es U \subseteq \real^m$ is a given set.
    By doing so, IR is implicitly characterized by the fact that the feasible set $\mathcal F(t)\coloneqq\{\hat u:B\hat u=\tau(t)\}\cap\mathcal U$ does not reduce to an one-point set, at least for some $t$. Subsequently, crucial questions are the followings: How to predict a priori that this condition holds? If one can prove that it does not, can this condition be achieved by redesigning the high level controller delivering $\tau(t)$? Or, by considering either a different reference signal to be tracked and/or a different initial condition? Answers to those questions are rather involved. Firstly, because the high level controller is usually assumed to be given and is therefore out of the scope of the analysis. Secondly, because of the interplay between the two levels of the controller.

    Some of the input-to-output methods also treat the input constraints. In \cite{Zaccarian2009}, the so-called allocator block injects an additional input signal which is invisible from the output and such that the resulting input vector remains within the saturation limits. Authors of \cite{Valmorbida2013,Galeani2013} offer methods for the computation of the zero-error steady-state solutions of the regulation problem, under input constraints. For the same problem, an online optimizer scheme is proposed in \cite{Galeani2015}, with the goal of promoting steady-state inputs with the smallest infinity norm. In the discrete-time case, a model predictive control scheme is offered in \cite{Zhou2016a}, to ensure that the internal dynamics comply with state and input constraints.
    All those contributions are valuable attempts toward control design handling input and state constraints. They share the following common pattern, though: First, IR is defined by referring to the \emph{unconstrained} context and, second, the control methodology is exposed by referring to input constraints. As a result, it cannot be ensured that the proposed methods apply to the class of IR systems characterized beforehand.

    This is in stark contrast with \cite{Kreiss2020} where IR is redefined as non injectivity of the input-to-output mapping, i.e. an output $y$ together with an initial state $x_0$ does not uniquely determine corresponding input $u$. Unlike the other definitions, this new formulation refers to signals with the aim of facilitating extensions out of the class of unconstrained linear time-invariant systems. This paper illustrates this potentiality, by considering not only input but also state constraints, i.e. both \eqref{eq:impconst} and the following condition hold for all $t\in\real_{\geq 0}$:
    \begin{equation}
        x(t)\in\es{X},
    \end{equation}
\end{subequations}
where $\es{X}\subseteq\real^n$ is a given set. In the sequel, definitions proposed in \cite{Kreiss2020} are applied \emph{verbatim} to the constrained context.

Then, an immediate question is to ask whether this non uniqueness of the input remains valid when considering other pairs $(x_0,y)$?
In the unconstrained context where $\mathcal{U}=\real^m$ and $\mathcal{X}=\real^n$ hold, the answer is positive, regardless of $(x_0,y)$ (see \cite[\locglobUnc]{Kreiss2020}). In the constrained context, the following example explicitly shows that the answer can be negative, so that constraints can challenge redundancy.

\begin{exmp}\label{ex:XUbound} 
    Consider the following system
    \[
        \dot x= - x + \begin{bmatrix} 1 & 1 \end{bmatrix} u,\quad y = x + \begin{bmatrix} 1 & 0 \end{bmatrix} u,
    \]
    for which each input is enforced to be non-negative, whereas state trajectory is unconstrained, i.e. $u(t)\in\mathcal U=\real_{\geq 0}^2$ and $x(t)\in\mathcal X=\real$
    for all $t\in\real_{\geq 0}$. Observe that input-to-output mapping is fully captured via the following equation on which $y$ acts as a parameter and $u\eqqcolon[u_a,u_b]^\intercal$ are the unknown functions:
    \begin{equation} \label{eq:XUbound}
        \dot y + y - u_b = \dot u_a + 2 u_a,
    \end{equation}
    with $u_a(0)=y(0)-x(0)$. Consider the following case study:
    \begin{enumerate}[(i)]
        \item Define $x_{0,1}\coloneqq -1$ and $y_1:t\mapsto -e^{-t}$. In this case, $u_1\eqqcolon[u_{1,a},u_{1,b}]^\intercal=\0$ is the unique input satisfying $u(t)\in\mathcal U$ leading to $y_1$ when $x(0)=x_{0,1}$. Indeed, \eqref{eq:XUbound} reduces to $- u_b = \dot u_a + 2 u_a$ with $u_a(0)=0$, so that strictly positive value of $u_b(t)$ induces violation of constraint $u_a(t)\geq 0$. Hence, $u_{1}(t)$ must equal $\0$ at all time.
        \item Keep $x(0)=x_{0,1}$ and choose $y_2=\0$. Clearly, either $u_2:t\mapsto [e^{-2t},0]^\intercal$ or $u_3:t\mapsto [e^{-3t},e^{-3t}]^\intercal$ produce $y_2$ and are compatible with the constraints.
        \item Define $x_{0,2}=0$ and let $y_2=\0$ be unchanged. For the same reason as in (i), output $y_2$ can only be produced by $u_1=\0$ for $x(0)=x_{0,2}$.
    \end{enumerate}
    To sum up, when $x(0)=x_{0,1}$, output $y_2$ can be {produced} by multiple inputs. On the contrary, substituting initial state by $x_{0,2}$ or output by $y_1$ makes the input unique. Therefore, the ability of designing different inputs giving rise to the same output depends on both output trajectory $y$ and initial condition~$x_0$.
\end{exmp}
This example shows that constraints give rise to unseen phenomena when instantaneous value of input and state vectors are free, as in \cite{Kreiss2020}.

From the above discussion, IR and input constraints are intrinsically related research fields. If the former is still at its infancy by many aspects, the latter benefits from solid results on classical notions of control theory. Together with stabilizability, the concept of controllability have probably monopolized most of the attention of the control community working on input constrained dynamical systems, see \cite{Brammer1972} and any standard textbook on this topic like e.g. \cite{Tarbouriech2011,Hu2001,Saberi2002}. Roughly speaking, this notion is related to the existence of a suitable input trajectory. On the contrary, IR deals with uniqueness of this input. The key point is that the question of existence received much more attention than the one of uniqueness, by far.

The bottom line of the previous discussion is that, to our best knowledge, IR has been neither defined nor characterized for systems subject to input or state constraints. In this paper, this challenge is tackled head-on. By doing so, it brings closer the literature dedicated to IR and the one devoted to input and state constraints.

Main contributions of this paper are now exposed.
1)~Framework introduced in \cite{Kreiss2020} is first enriched by new definitions to handle the more complex phenomenology on IR in the constrained context, as partially highlighted by Ex.~\ref{ex:XUbound}.
2)~Then, a sufficient condition for $(x_0,y)$ to be compatible with distinct inputs is derived, see Th.~\ref{th:singNL}. This can be considered as the main achievement of this paper. To arrive at this result, a rigorous incremental point of view is adopted. Specifically, the inputs $u_1$ and $u_2$ leading to the same output are described as $u_1$ and $u_1+\tilde u$ with $\tilde u = u_2-u_1$, respectively. Let us emphasize that this analysis is conducted for \emph{arbitrary} input and state constraint sets $\mathcal{U}$ and $\mathcal{X}$ (e.g. they can be neither convex nor connected).
3)~The case where $\mathcal{U}$ and $\mathcal{X}$ are linear is then treated as a byproduct of this analysis. In this case, a comprehensive characterization of IR as well as its taxonomy is derived. The concept of degree of IR, as defined in \cite{Kreiss2020}, is also generalized.
In a nutshell, the overall contribution of this paper is an extensive discussion on how $\mathcal{U}$ and $\mathcal{X}$ impact properties of IR associated with the corresponding unconstrained system.

Section~II sets the stage of this study. Section~III enriches conceptual framework on IR introduced in \cite{Kreiss2020}.
Main results are offered in Section~IV, in the general context where $\mathcal{U}$ and $\mathcal{X}$ are arbitrary.
Section~V focuses on the particular case where input and state constraints are linear.

\paragraph*{Notations}
Symbols $\wedge$, $\vee$ and $\lnot$ stand for logical operators ``and'', ``or'' and ``not'', respectively. Symbol $\0$ stands for anything that is not a real number and is zero (a vector, matrix, map, or subspace), according to context. 
Identity matrix is denoted by $\I$.
Set $\sigma (A)$ is the spectrum of square matrix $A$.
Cardinality of a set $S$, denoted by $\card{S}$, is said to be greater than or equal to $\alpha\in\mathbb N$ if $S$ admits at least $\alpha$ distinct members. 
In particular, if $\fs{S}\subseteq\{u:\real\rightarrow\real^n\}$ is a function space, then $\card{\fs{S}}>\alpha$ means that there exist $u_1,\ldots,u_\alpha\in\fs{S}$ such that every $u_i,u_j$ differ on a strictly positive measure set, i.e. \mbox{$\textstyle\int{\left\Vert u_i(\tau)-u_j(\tau) \right\Vert d\tau}>0$} for all $i\neq j\in\{1,\ldots,\alpha\}$. Signals $u_1, u_2$ are equal if $\card{\{u_1,u_2\}}=1$ so that $u_1(t)= u_2(t)$ holds for almost all $t$.
Subsets (not necessarily vector/linear spaces) of {Euclidean} space are denoted by script symbols, e.g. $\mathcal U$ or $\mathcal X$. {$\text{int}(\es{X})$} refers to the interior of $\es{X}$. Apart from $\I$, most of the time capital bold letters refer to function space like e.g. $\fs{U}=\{u:\real_{\geq 0}\rightarrow \mathcal U\}$. 
Denote by $\fs{C}(I,\mathcal C)$ (resp. $\fs{PC}(I,\mathcal C)$) 
the set of continuous functions (resp. piecewise continuous functions) 
from $I\subseteq\real$ to $\mathcal C$. 
Set $S-a$ reads $\{s-a:s\in S\}$, so that $a\in S$ is equivalent to $\0\in S-a$. Given set $\mathcal A\subseteq\real^n$ and matrix $B$ (not necessarily invertible or square) with $n$ lines, set $B^{-1}\mathcal A$ reads $\{u:Bu\in\mathcal A\}$.
With slight abuse of notation, relationship $x\in \mathcal X$, where $x(\cdot)$ is a function of time, means $x(t)\in \mathcal X$ for all $t$ in the domain of $x(\cdot)$. Laplace transform of $u$ is denoted by $s\mapsto \mathfrak L[u](s)$.

\section{Context of the study}

This section inherits and enriches the notations and the framework of \cite{Kreiss2020}. We refer the reader to \cite{Kreiss2020} for details, bearing in mind that sets $\fs U$, $\fs Q$ and $\fs W$ of \cite{Kreiss2020} are here renamed as $\fs U_\Sigma$, $\fs Q_\Sigma$ and $\fs W_\Sigma$, respectively.

\subsection{Linear system with constraints}

From \eqref{eq:sysAB}, the input-to-state relationship is concisely captured via $\H_x[x_0;\cdot]$ which maps an input trajectory $u(\cdot)$ to the state trajectory $x(\cdot)$ produced by the system when excited by $u(\cdot)$ with an initial condition $x(0)=x_0\in \real^n$. The input-to-output mapping $\H[x_0;\cdot]:u\mapsto C\H_x[x_0;u]+Du$ derives from \eqref{eq:sysCD}.

Throughout this paper, inputs $u$ are assumed to belong to $\fs{U}_\Sigma$, the set of causal, piecewise continuous and exponentially bounded signals.\footnote{By causal signal, we mean a signal which is zero for all strictly negative time instant.} This ensures that the Laplace transform of $u$, as well as that of corresponding $x=\H_x[x_0;u]$ and $y=\H[x_0;u]$, exist, whatever is $x_0$. Since $x$ and $y$ are continuous and piecewise continuous, respectively, one can define $\fs{X}_\Sigma\coloneqq\fs{C}(\real_{\geq 0},\real^n)$ and $\fs{Y}\coloneqq \fs{PC}(\real_{\geq 0},\real^p)$ as the codomain of $\H_x[x_0;\cdot]$ and $\H[x_0;\cdot]$, respectively.

Let $x_0\in\real^n$ be a given initial condition. The set of all triples $(u,x,y)$ (resp. pairs $(u,y)$) \emph{compatible for $x_0$} is denoted by $\fs{Q}_\Sigma(x_0)$ (resp. $\fs{W}_\Sigma (x_0)$), i.e.
\begin{align*}
    \fs{Q}_\Sigma (x_0) & \coloneqq  \{(u,x, y)\mid\H_x[x_0;u]= x,\; \H[x_0;u]=y\},            \\
    \fs{W}_\Sigma (x_0) & \coloneqq  \{(u,y)\mid \exists x : (u,x, y)\in\fs{Q}_\Sigma (x_0)\}.
\end{align*}

\ifext{\color{blue}

    Let us already introduce Fig.~\ref{fig:graph1} supporting forthcoming definitions and discussions. This picture should be regarded as a ``running'' informal sketch the reader can refer to, to understand concepts introduced throughout the rest of this section. For didactical purpose, assume that every triple compatible with $\Sigma$ is depicted by Fig.~\ref{fig:graph1}, so that $\fs{Q}_\Sigma(x_0)$ has finite cardinality. The reader is also referred to Tab.~\ref{tab:sets}, where main notations are collected together. Note that some of them are introduced in the sequel.

    \begin{figure}
        \centering
        \includegraphics[width=.8\columnwidth]{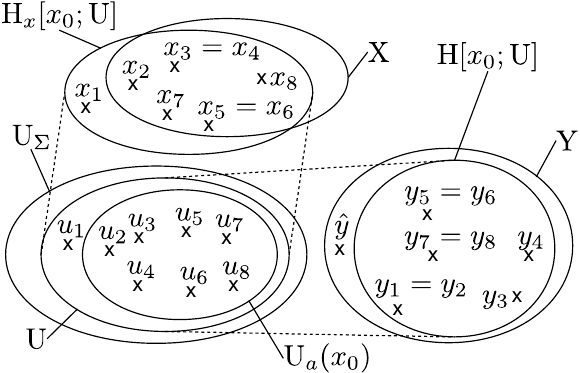}
        \caption{Graphical illustration of input/state/output triples associated with $\Sigma$ and $\mathfrak S$ for some fixed $x_0\in\es{X}$. It holds $\fs{Q}_\Sigma(x_0)=\{(u_i,x_i,y_i)_{i\in\{1,\ldots,8\}}\}$.} \label{fig:graph1}
    \end{figure}

}\fi

In contrast with \cite{Kreiss2020}, this paper deals with dynamical system $\mathfrak S$ deriving from $\Sigma$ by imposing input and state constraints. For this reason, input and state signals of $\mathfrak S$ belong to sets $\fs{U}$ and $\fs{X}$ defined as follows:
\begin{align*}
    \fs{U} & \coloneqq  \{u:\real_{\geq 0}\rightarrow \es{U}\} \cap \fs{U}_\Sigma, \\
    \fs{X} & \coloneqq  \{x:\real_{\geq 0}\rightarrow \es{X}\}\cap \fs{X}_\Sigma.  
\end{align*}
As a result, the set $\fs{Q}(x_0)$ of triple $(u,x,y)$ compatible with $\mathfrak S$ derives from $\fs{Q}_\Sigma(x_0)$ by excluding trajectories which violate constraints, i.e.
\begin{equation} \label{eq:Qdef}
    \fs{Q}(x_0) \coloneqq  \fs{Q}_\Sigma(x_0) \cap (\fs{U} \times \fs{X} \times \fs{Y}).
\end{equation}
Thus, $(u,x,y)\in\fs{Q}(x_0)$ implies that $(u(t),x(t))\in\es{U}\times\es{X}$ holds for all non-negative time~$t$. Similarly, denote by $\fs{W}(x_0)$ the set of input-to-output pair compatible for $x_0$ with $\mathfrak S$, i.e. $(u,y)\in\fs{W}(x_0)$ if there exists $x$ such that $(u,x,y)\in\fs{Q}(x_0)$.

\ifext{\color{blue}
    \begin{table}
        \centering
        \caption{\label{tab:sets} Nomenclature}
        \begin{tabular}{|c|c|c|}
            \hline
                                         & $\Sigma$               & $\mathfrak{S}$              \\ \hline
            Input and state constraints  & $\real^m\times\real^n$ & $\es{U}\times\es{X}$        \\
            Input-output pairs           & $\fs{W}_\Sigma(x_0)$   & $\fs{W}(x_0)$               \\
            Input-state-output triples   & $\fs{Q}_\Sigma(x_0)$   & $\fs{Q}(x_0)$               \\
            Admissible inputs            & $\fs{U}_\Sigma$        & $\fs{U}_a(x_0)$             \\
            Admissible initial condition & $\real^n$              & $\es{X}_c$                  \\
            Admissible pairs $(x_0,y)$   & $\es{X}\times\fs{Y}$   & $\fs{A}$                    \\
            Input-to-output mapping      & $\H[x_0;\cdot]$        & $\mathfrak{H}[x_0;\cdot]$   \\
            Input-to-state mapping       & $\H_x[x_0;\cdot]$      & $\mathfrak{H}_x[x_0;\cdot]$ \\\hline
        \end{tabular}
    \end{table}
}\fi

\subsection{Admissible signals}

In general, constraints \eqref{eq:impstatconst} induce specific difficulties, like the one highlighted by the following example.
\begin{exmp}\label{ex:Xc}
    Consider the dynamical system $\dot x = x+u$ with $\mathcal X=\mathcal U=[0,1]$. Unique solution reads $x:t\mapsto x_0 e^t + \int^t_0 e^{t-\tau}u(\tau) \mathrm{d}\tau$. Therefore, if $0<x_0\leq 1$ holds, then $x(t)$ will inevitably escape from $\mathcal X$, whatever is $u\in\fs{U}$. Hence, there is no $u\in\fs{U}$ such that $x\in\fs{X}$ holds. On the contrary, selecting $x_0=0$ and $u=\0$ yields $x=\0\in\fs{X}$.
\end{exmp}
As illustrated by the previous example, depending on initial state $x_0$ and set $\es{U}\times\es{X}$, it might happen that no state trajectory lying identically in $\es{X}$ exists, whatever is $u\in\fs{U}$. For this reason, denote by $\mathcal X_c$ the set of all initial conditions $x_0$ for which there exists at least one compatible input-to-output pair: 
\begin{equation*}
    \mathcal X_c \coloneqq  \{x_0\in\real^n: \card{\fs{W} (x_0) }\geq 1\} \subseteq \mathcal X.
\end{equation*}
For Ex.~\ref{ex:Xc}, $\mathcal X_c$ equals $\{\0\}$, a proper subset of $\mathcal X$.
\ifext{\color{blue}
    Note that $\mathcal X_c\subseteq \mathcal X$ necessarily holds since $x_0\notin \mathcal X$ will violate inclusion $x(t)\in \mathcal X$ for $t=0$.
}\fi

An input signal $u$ (resp. output signal $y$) is said \emph{admissible for $x_0$} with $\mathfrak S$ if there exists an output $y$ (resp. an input $u$) for which $(u,y)$ is compatible for $x_0$ with $\mathfrak S$.
Let $\fs{U}_a(x_0)\subseteq \fs{U}_\Sigma$ be the set of admissible input trajectories for $x_0$ with $\mathfrak S$, i.e. those which comply with input \emph{and} state constraints:
\begin{equation*}
    \fs{U}_a(x_0) \coloneqq  \{u\in\fs{U} \mid \exists y:(u,y)\in\fs{W}(x_0)\}.
\end{equation*}
Clearly, $\fs{U}_a(x_0)$ is non empty if and only if (iff) $x_0$ belongs to $\mathcal X_c$. For Ex.~\ref{ex:Xc}, $\fs{U}_a(x_0)$ equals $\varnothing$ (resp. $\{\0\}$) for $x_0\in]0,1]$ (resp. $x_0=0$, since every input distinct from $\0$ will drive $x(t)$ out of $\es{X}$).

Let $\fs{A}$ be the set of \emph{admissible pair} $(x_0,y)\in\es X \times \fs Y$ for which $y$ is admissible for $x_0$, i.e.
\[
    \fs A \coloneqq  \{ (x_0,y) \in \es{X}\times\fs{Y} \mid \exists u: (u,y)\in\fs W(x_0)\}.
\]
Unless otherwise specified, compatibility and admissibility shall be understood ``\emph{with} $\mathfrak S$''.

Finally, input-to-state mapping $\mathfrak{H}_x[x_0;\cdot]$ (resp. input-to-output mapping $\mathfrak{H}[x_0;\cdot]$) associated with $\mathfrak S$ and parametrized by $x_0$, is naturally defined as the restriction of $\H_x[x_0;\cdot]$ (resp. $\H[x_0;\cdot]$) to $\fs{U}_a(x_0)$, i.e.
\begin{align*}
    \mathfrak{H}_x[x_0;\cdot] & \coloneqq  \H_x[x_0;\cdot] \mid \fs{U}_a(x_0), \\
    \mathfrak{H}[x_0;\cdot]   & \coloneqq  \H[x_0;\cdot] \mid \fs{U}_a(x_0).
\end{align*}

\ifext{\color{blue}

    \subsection{Graphical illustration of input/state/output mappings}

    First comments about Fig.~\ref{fig:graph1} are as follows. (i) $\fs{U}$ (resp. $\fs{U}_a(x_0)$) derives from $\fs{U}_\Sigma$ by ruling out signals which do not comply with input constraints (resp. input and state constraints), so that $\fs{U}_a(x_0)\subseteq \fs{U}\subseteq \fs{U}_\Sigma$ holds regardless of $x_0\in\es{X}$. (ii) It holds $(u_i,x_i,y_i)\in\fs{Q}(x_0)$ for all $i$ but $i=1$ since $x_1\not\in \fs{X}$ which means that state constraints are violated when exerting input $u_1$. (iii) For the same reason, the image of $\fs{U}$ by $\H_x[x_0;\cdot]$ is not contained in $\fs{X}$ in general. (iv) $\H[x_0;\fs{U}]$ is, in general, a proper subset of $\fs{Y}$, as illustrated by the existence of $\hat y\in\fs{Y}\setminus \H[x_0;\fs{U}]$, i.e. there is no $u$ such that $(u,\hat y)\in\fs{W}(x_0)$. Existence of $\hat y$ is possible since $\Sigma$ is not assumed to be right-invertible and input space $\fs{U}_\Sigma$ does not contain distributions (see e.g. \cite[Chap. 8]{Trentelman2001}). (v) For analogous reasons, $\H_x[x_0;\fs{U}]$ does not contain $\fs{X}$ in general. (vi) Inequality $y_3\neq y_4$ holds even if $x_3$ equals $x_4$. This situation can be encountered, whenever $\Sigma$ is proper but not strictly proper.

}\fi

\section{Enriched conceptual framework on IR}

In this section, we enrich definitions and taxonomy associated to IR proposed in \cite{Kreiss2020}, in order to cope with the more delicate context of constrained dynamics where redundancy depends on both the output $y$ and the initial state $x_0$, see Ex.~\ref{ex:XUbound}. 

\subsection{\Sing{} pairs}

\begin{defn}[\Sing{} pair]
    A pair $(x_0,y)\in \mathcal X \times \fs{Y}$ is \emph{\sing{}} if output $y$ can be produced by (at least) two distinct inputs for initial condition $x_0$, i.e.
    \begin{equation}\label{eq:singPair}
        \begin{array}{r}
            \exists (u_1,y_1),(u_2,y_2)\in\fs{W}(x_0) :\hspace*{6em} \\
            u_1\neq u_2,\; y_1= y_2=y.
        \end{array}
    \end{equation}
\end{defn}

Let $\fs S\subseteq \es{X}\times\fs{Y}$ be the set of \sing{} pairs, i.e.
\begin{equation*}
    \normalfont{\bf S} \coloneqq  \{ (x_0,y) \in \es{X}\times\fs{Y}  : \eqref{eq:singPair}\}.
\end{equation*}
The following inclusion chain holds:
\begin{equation*}
    \fs S \; \subseteq \; \fs A \; \subseteq \; \es{X}\times\fs{Y}.
\end{equation*}
It suggests (i) that $\es{X}\times\fs{Y}$ might contain non admissible pairs $(x_0,y)$ (in this case, $\fs A$ is a proper subset of $\es{X}\times\fs{Y}$) and (ii) that an admissible pair $(x_0,y)\in\fs A$ might be non \sing{} (in this case, $\fs S$ is a proper subset of $\fs{A}$).  Finally, inclusion $\fs S \subseteq \fs A$ means that every \sing{} pair $(x_0,y)$ is admissible.

Let us now introduce state trajectories into \eqref{eq:singPair}: Given a pair $(x_0,y)\in \mathcal X \times \fs{Y}$, define the following relationships
\begin{align} \label{eq:kind1}
    \begin{array}{r}
        \exists (u_1,x_1,y_1),(u_2,x_2,y_2)\in\fs{Q}(x_0) :\hspace*{4em} \\
        u_1\neq u_2,x_1= x_2, y_1= y_2=y,
    \end{array} \\\label{eq:kind2}
    \begin{array}{r}
        \exists (u_1,x_1,y_1),(u_2,x_2,y_2)\in\fs{Q}(x_0) :\hspace*{4em} \\
        u_1\neq u_2,x_1\neq x_2, y_1= y_2=y.
    \end{array}
\end{align}
Observe that if $(x_0,y)$ is \sing{} then at least one of the above relationships holds.

Making use of \eqref{eq:kind1} and \eqref{eq:kind2}, instead of \eqref{eq:singPair}, allows for investigating origin of \singy{}. This gives rise to a taxonomy distinguishing \sing{} pairs.
\begin{defn}[\Sing{} pair of the $k$-th kind] \label{prop:pairkind}
    Let $(x_0,y)\in \mathcal X \times \fs{Y}$ be an \sing{} pair. It is of:
    \begin{itemize}
        \item The \emph{1st kind} if \eqref{eq:kind1} holds but \eqref{eq:kind2} does not, i.e. if the following implication
              \begin{equation} \label{eq:kind1prop}
                  \left.\begin{array}{r}
                      u_1\neq u_2 \\
                      y = y_1 = y_2
                  \end{array}\right\}
                  \Rightarrow
                  x_1= x_2
              \end{equation}
              holds for all $(u_1,x_1,y_1),(u_2,x_2,y_2)\in\fs{Q}(x_0)$;
        \item The \emph{2nd kind} if \eqref{eq:kind1} does not hold but \eqref{eq:kind2} does, i.e. the following implication
              \begin{equation} \label{eq:kind2prop}
                  \left.\begin{array}{r}
                      u_1\neq u_2 \\
                      y = y_1 = y_2
                  \end{array}\right\}
                  \Rightarrow
                  x_1 \neq x_2
              \end{equation}
              holds for all $(u_1,x_1,y_1),(u_2,x_2,y_2)\in\fs{Q}(x_0)$;
        \item The \emph{3rd kind} if both \eqref{eq:kind1} and \eqref{eq:kind2} hold, i.e. neither \eqref{eq:kind1prop} nor \eqref{eq:kind2prop} is valid for all $(u_1,x_1,y_1),(u_2,x_2,y_2)\in\fs{Q}(x_0)$.
    \end{itemize}
\end{defn}

Previous definition induces that different kinds are mutually exclusive: No pair $(x_0,y)$ can be simultaneously of different kinds. This motivates partitioning of $\fs S$ as follows:
\begin{equation*}
    \fs{S} = \fs S_1 \cup \fs S_2 \cup \fs S_3,
\end{equation*}
where $\fs S_k$ gathers \sing{} pairs of the $k$-th kind and satisfies $\fs S_k\cap \fs S_j=\varnothing$ for all $k\neq j$.

\ifext{\color{blue}
    Let us provide a didactical illustration of those concepts via Fig.~\ref{fig:graph1}. Observe that the set $\fs{S}$ of all \sing{} pairs of $\mathfrak S$ equals $\{(x_0,y_5),(x_0,y_7)\}$. It also comes out $\fs{S}_1=\{(x_0,y_5)\}$ and $\fs{S}_2=\{(x_0,y_7)\}$. Note that pair $(x_0,y_1)$ is not \sing{}, even if $y_1=y_2$ and $u_1\neq u_2$ hold. Indeed, $x_1\not\in\fs{X}$ holds so that $(u_1,y_1)$ does not belong to $\fs{W}(x_0)$.

    \begin{rem}[Alternative formulation of \eqref{eq:singPair}]
        Let $\mathfrak{H}^{-1}[x_0;y]$ refers to preimage of $y\in \fs{Y}$ by $\mathfrak{H}[x_0;\cdot]$, i.e.
        \begin{multline}\label{eq:Hinv}
            \mathfrak{H}^{-1}[x_0;y] \coloneqq\{u\in\fs{U} : y=\mathfrak{H}[x_0;u]\} \\
            = \{u\in\fs{U}:(u,y)\in\fs{W}(x_0)\}.
        \end{multline}
        The following relationship
        \begin{equation}
            \card{\mathfrak{H}^{-1}[x_0;y]} \geq 2
        \end{equation}
        gives an alternative way to define \singy{} of $(x_0,y)\in \mathcal X \times \fs{Y}$.
        This reformulation highlights the fact that existence of \singy{} is equivalent to \emph{non-injectivity} of $\mathfrak{H}[x_0;\cdot]$ for some $x_0$.
    \end{rem}
}\fi

\subsection{IR systems as defined in \cite{Kreiss2020}}

Definitions and taxonomy introduced in \cite[Sec.~2]{Kreiss2020} apply \emph{verbatim} for $\mathfrak S$.

\begin{defn}[IR]
    System $\mathfrak{S}$ is \emph{IR} if there exists an \sing{} pair $(x_0,y)\in \mathcal X \times \fs{Y}$, i.e. $\fs S\neq\varnothing$.
\end{defn}

Whenever every \sing{} pair enjoys the same qualifying term, system $\mathfrak{S}$ itself shall inherit from the proposed taxonomy, previously coins for \sing{} pairs.

\begin{defn}[IR of the $k$-th kind]
    System $\mathfrak{S}$ is \emph{IR of the $k$-th kind} if (i) it is IR and (ii) every \sing{} pair $(x_0,y)\in \mathcal X\times \fs{Y}$ 
    is of the $k$-th kind, i.e. \mbox{$\fs S_k=\fs S\neq \varnothing$}.
\end{defn}

\subsection{Uniform IR}

To prove that $\mathfrak{S}$ is IR, it suffices to find a single pair $(x_0,y)$ which is \sing{}. Yet, for some system, this \singy{} might occur \emph{for all} admissible pair $(x_0,y)\in(\es X\times\fs Y)$. The adverb ``uniformly'' is used to qualify this situation.

\begin{defn}[Uniform IR]
    System $\mathfrak{S}$ is \emph{uniformly} IR if (i) it is IR and (ii) every admissible pair $(x_0,y)\in(\es X\times\fs Y)$ is \sing{}, i.e. $\fs A = \fs S\neq \varnothing$ holds.
\end{defn}

This property of uniformity can be combined with the proposed taxonomy of system $\mathfrak{S}$.

\begin{defn}[Uniform IR of the $k$-th kind]
    System $\mathfrak{S}$ is \emph{uniformly IR of the $k$-th kind} if (i) it is uniformly IR and (ii) every \sing{} pair $(x_0,y)\in \mathcal X\times \fs{Y}$ 
    is of the $k$-th kind, i.e. $\fs A=\fs S=\fs S_k\neq \varnothing$.
\end{defn}

\ifext{\color{blue}
    In view of Fig.~\ref{fig:graph1},  $\mathfrak S$ is neither IR of any kind (both $\fs{S}_1$ and $\fs{S}_2$ are non empty) nor uniformly IR ($\fs{A}\ni(x_0,y_2)\not\in\fs{S}$).

    Tab.~\ref{tab:def} can be referred to as a summary of the introduced definitions.

    \begin{table}
        \centering
        \caption{\label{tab:def} Definitions associated to IR}
        \begin{tabular}{|c|c|}
            \hline
            IR                            & $\fs S\neq\varnothing$                \\
            IR of the $k$-th kind         & $\fs S_k=\fs S\neq \varnothing$       \\
            Uniform IR                    & $\fs A = \fs S\neq \varnothing$       \\
            Uniform IR of the $k$-th kind & $\fs A=\fs S=\fs S_k\neq \varnothing$ \\ \hline
        \end{tabular}
    \end{table}
}\fi

\section{General results \label{sec:NL}}

In addition of being part of the description of $\mathfrak{S}$, system $\Sigma$ can be regarded as the unconstrained version of $\mathfrak{S}$. Therefore, comparison between those two systems allows to evaluate how constraints impact IR. This is the goal of this section.

Note that most of the proofs of this section and the next one are postponed in the appendices. Among them, Appendix~\ref{ap:incAn} proposes a key analysis of input-to-output mapping associated with $\mathfrak{S}$ by adopting an incremental view point. This analysis is instrumental in achieving tractable characterization of definitions associated with IR. 

\subsection{Sufficient condition for \singy{} to be preserved} 

Via Ex.~\ref{ex:XUbound}, it has been observed that constraints can destroy \singy{} of a pair $(x_0,y)$, in the sense that $(x_0,y)$ can be \sing{} for $\Sigma$ but not for $\mathfrak{S}$. Following theorem, proved in Appendix~\ref{ap:incAn}, offers a sufficient condition to prevent this situation to occur, i.e. for \singy{} to be preserved in the constrained context.

\begin{thm} \label{th:singNL}
    Define the following integer:
    \begin{equation*}
        \rho \coloneqq \dim (\ker{B}\cap\ker{D}),
    \end{equation*}
    as well as the following relationships:
    \begin{subequations}
        \begin{align} \label{eq:fenetreU}
            u(t) & \in \mathrm{int}(\mathcal U), \\ \label{eq:fenetreX}
            x(t) & \in \mathrm{int}(\mathcal X),
        \end{align}
    \end{subequations}
    parametrized by $t\in\real_{\geq 0}$, $u\in\fs{U}$ and $x\in\fs{X}$. Assume that the unconstrained system $\Sigma$ is IR. Then, admissible pair $(x_0,y)\in\fs{A}$ is IR (for $\mathfrak{S}$) if there exist $(u,x)$ and $0\leq t_0<t_f$ such that $(u,x,y)\in\fs{Q}(x_0)$ holds and:
    \begin{itemize}
        \item if $\rho>0$, \eqref{eq:fenetreU} holds for all $t\in]t_0,t_f[$;
        \item if $\rho=0$, \eqref{eq:fenetreU} and \eqref{eq:fenetreX} hold for all $t\in]t_0,t_f[$.
    \end{itemize}
\end{thm}

When $\es{U}$ is open, \eqref{eq:fenetreU} (resp. \eqref{eq:fenetreX}) is valid for all $u\in\fs{U}$ (resp. for all $x\in\fs{X}$) and for all $t\in\real_{\geq 0}$. 
This leads to the following corollary.

\begin{cor}\label{cor:open}
    Assume that the unconstrained system $\Sigma$ is IR. If one of the following conditions is valid:
    \begin{itemize}
        \item $\rho>0$ and $\es{U}$ is open,
        \item $\rho=0$ and both $\es{U}$ and $\es{X}$ are open, 
    \end{itemize}
    then, $\mathfrak S$ is uniformly IR, i.e. \mbox{$\fs{S}=\fs{A}=\varnothing$}.
\end{cor}


Contraposition of Th.~\ref{th:singNL} is also of major importance. In the context of Th.~\ref{th:singNL}, assume that $\rho>0$ (resp. $\rho=0$) holds and that $(x_0,y)$ is admissible but \emph{not} \sing{}, so that there is a unique pair $(u,x)$ satisfying $(u,x,y)\in\fs{Q}(x_0)$. Then for all non empty $]t_0,t_f[\subseteq\real_{\geq 0}$, condition \eqref{eq:fenetreU} (resp. \eqref{eq:fenetreU} and \eqref{eq:fenetreX}) is violated for some $t=t_i\in]t_0,t_f[$, i.e. $u(t_i)\in\es{U}\setminus\text{int}(\es{U})$ holds (resp. either $u(t_i)\in\es{U}\setminus\text{int}(\es{U})$ or $x(t_i)\in \es{X}\setminus\text{int}(\es{X})$ hold). As formally proved in Appendix~\ref{ap:notFenetre}, this observation together with piecewise continuity of $u$ and continuity of $x$ lead to the following corollary.

\begin{cor}\label{cor:notFenetre}
    Define the following relationships:
    \begin{subequations}
        \begin{align} \label{eq:notfenetreU}
            u(t) & \in \es U\setminus \mathrm{int}(\es U),      \\ \label{eq:notfenetreX}
            x(t) & \in \es X\setminus \mathrm{int}(\mathcal X),
        \end{align}
    \end{subequations}
    parametrized by $t\in\real_{\geq 0}$, $u\in\fs{U}$ and $x\in\fs{X}$. Assume that the unconstrained system $\Sigma$ is IR.
    Then, admissible pair $(x_0,y)\in\fs{A}$ \emph{is not IR (for $\mathfrak{S}$) only if} the unique $(u,x)$ satisfying $(u,x,y)\in\fs{Q}(x_0)$ is such that:
    \begin{itemize}
        \item if $\rho>0$, then \eqref{eq:notfenetreU} holds for all $t\in\es{C}$;
        \item if $\rho=0$, then \eqref{eq:notfenetreU} or \eqref{eq:notfenetreX} holds for all $t\in\es{C}$.
    \end{itemize}
    Here, $\es{C}$ is the largest open subset of $\real_{\geq 0}$ where $u\in\es{U}$ is continuous.
\end{cor}

If $\es{U}$ and $\es{X}$ are closed and $\Sigma$ is IR, then Cor.~\ref{cor:notFenetre} proves that admissible pairs $(x_0,y)$ which are not \sing{} are necessarily associated with input and state trajectories that lie alternatively on the boundary of $\es{U}$ and $\es{X}$ for all $t$ where $u\in\es{U}$ is continuous.

\begin{exmp}\label{ex:bucks}
    Consider the parallel interconnection of two buck converters feeding a single resistive load of magnitude $R$ and fed by a single voltage source of magnitude $V$.
    Interested reader is referred to \cite{Tregoueet2019} for details about this application.

    This system can be controlled via a pulse width modulation strategy, so that the duty cycles of each converter are the components of the input vectors $u(t)\in\es{U}=[0,1]^2$. Resulting averaged dynamics can be modeled via \eqref{eq:sys} with
    \begin{equation*}
        A = \begin{bmatrix}
            0           & 0           & -\frac{1}{L}  \\
            0           & 0           & -\frac{1}{L}  \\
            \frac{1}{C} & \frac{1}{C} & -\frac{1}{RC}
        \end{bmatrix},
        B = \frac{V}{L} \begin{bmatrix}
            1 & 0 \\
            0 & 1 \\
            0 & 0
        \end{bmatrix},
        C = \begin{bmatrix} 0 \\ 0 \\ 1 \end{bmatrix}^\intercal,
        D = \0,
    \end{equation*}
    where $L$ and $C$ denote the inductances and the output capacitor magnitudes, respectively.
    Components of the state vector $x(t)\in\es{X}=\real^3$ are the two currents flowing through the inductances, followed by the voltage at the load. This voltage is the output signal.

    For this example, $\rho$ equals $0$ and one can prove that $\Sigma$ is IR, see \cite{Kreiss2018} where it is shown that $\es{R}(\Sigma)$ equals $\Span{\smallmat{1 & -1 & 0 }^\intercal}$.
    To demonstrate that $\mathfrak{S}$ is also IR, first pick any $x_0\in\es{X}=\real^3$. Then, define input $u_1$ as follows:
    \begin{equation*}
        u_1:t\mapsto\left\{\begin{array}{ll}
            \begin{bmatrix} 1-t & t \end{bmatrix}^\intercal,   & (t\in [0,1]), \\
            \begin{bmatrix} 0 & 1 \end{bmatrix}^\intercal  , & (t>1).
        \end{array}\right.
    \end{equation*}
    Let output $y_1$ satisfy $(u_1,y_1)\in\fs W(x_0)$.
    Since $u_1(t)$ belongs to $\mathrm{int}(\es U)$ for all $t\in]0,1[$ and $\mathrm{int}(\es X)=\es X$ so that \eqref{eq:fenetreX} holds for all $t\geq 0$, Th.~\ref{th:singNL} ensures that $(x_0,y_1)$ is an IR pair for $\mathfrak{S}$. Indeed, input $u_2:t\mapsto \smallmat{1\\0}$ leads to output $y_2=\H[x_0;u_2]$ which is equal to $y_1$. This fact can be highlighted via the following equation capturing the input-output mapping:
    \begin{equation}\label{eq:buckIO}
        C \ddot y + \frac{1}{R} \dot y + \frac{2}{L} y = \frac{E}{L} \begin{bmatrix} 1 & 1 \end{bmatrix} u.
    \end{equation}
    Indeed, observe that exerting either $u_1$ or $u_2$ lead to identical right-hand side of \eqref{eq:buckIO}, so that resulting outputs $y_{1,2}$ are the same. More generally, observe that a given output trajectory $y$ imposes the expression of $\smallmat{1 & 1} u$ via \eqref{eq:buckIO}, so that any signal of the form $\tilde u=\smallmat{1\\-1}w$ with $w:\real_{\geq 0}\rightarrow\real$ can be added to the input $u$ leading to $y$ without affecting this output, i.e. $\H[x_0;u+\tilde u]=y$.

    Consider now the input $u_3=\0$ leading to the output $y_3=\0$ for zero initial state $x_0=\0$. 
    The pair $(x_0,y_3)$ is \emph{not} IR. To see this, observe that \eqref{eq:buckIO} together with $(x_0,y)=(\0,\0)$ imply $\smallmat{1 & 1} u=0$. This last equation admits a unique solution lying identically in $\es U$, that is $u_3=\0$. Note that the fact that $(x_0,y_3)$ is not IR agree with the statement of Th.~\ref{th:singNL}: \eqref{eq:fenetreU} is never valid, whatever is the selection of the interval $]t_0,t_f[\subseteq\real_{\geq 0}$.
\end{exmp}

\begin{exmp}[Ex.~\ref{ex:XUbound} continued]
    Case (ii) proves that $\mathfrak S$ and, in turn, $\Sigma$ are IR. For cases (i) and (iii), unique input $u_1=\0$ making pairs $(x_{0,1},y_1),(x_{0,2},y_2)$ admissible satisfies $u_1(t)\in\es U\setminus \mathrm{int}(\es U)=(\{0\}\oplus\real_{\geq 0}) \cup (\real_{\geq 0}\oplus\{0\})$ for all $\es{C}=\real_{\geq 0}$, as predicted by Cor.~\ref{cor:notFenetre}.
\end{exmp}

\begin{rem}
    Assume that $\Sigma$ is IR. Clearly, Cor.~\ref{cor:open} ensures that $\mathfrak S$ is uniformly IR if both $\es{U}$ and $\es{X}$ are open, \emph{regardless of $\rho$}. In the same spirit, Cor.~\ref{cor:notFenetre} proves that $(x_0,y)\in\fs{A}\setminus\fs{S}$ holds only if the unique $(u,x)$ satisfying $(u,x,y)\in\fs{Q}(x_0)$ is such that \eqref{eq:notfenetreU} or \eqref{eq:notfenetreX} holds for all $t\in\es{C}$, \emph{whatever is the value of $\rho$}. Similarly, if there exist $(u,x,y)\in\fs{Q}(x_0)$ and $0\leq t_0<t_f$ such that both \eqref{eq:fenetreU} and \eqref{eq:fenetreX} hold for all $t\in]t_0,t_f[$, then $(x_0,y)$ is \sing{} for $\mathfrak S$ \emph{regardless of $\rho$}, by virtue of Th.~\ref{th:singNL}.
    %
    %
\end{rem}

%
%
%
%

\subsection{How constraints impact the kind of \singy{} ?}

Assume that the unconstrained system $\Sigma$ is \sing{} of the $k$-th kind. Given an \sing{} pair $(x_0,y)\in\fs{S}$. Then, a natural question is to ask: Does $(x_0,y)$ belong to $\fs{S}_k$ ? Saying it differently, do the constraints can modify the kind of \singy{}  of a pair ?

For $k=3$, the answer to the last question might be positive even when the constraints are linear, as shown by the next example.

\begin{exmp}\label{ex:chgkind}
    Define dynamical system $\Sigma$ as follows
    \begin{align*}
        \dot x(t) & =  \begin{bmatrix}-1&0&0\\0&-1&0\\0&0&-1\end{bmatrix}x(t) + \begin{bmatrix}	1 & 0 & 0 & 0\\ 0 & 1 & 0 & 0\\ 0 & 0 & 1 & 1\end{bmatrix} u(t) , \\
        y (t)     & = \begin{bmatrix}0&0&1\end{bmatrix}x(t)
    \end{align*}
    In view of \cite[\characIR]{Kreiss2020}, one can prove that $\Sigma$ is IR of the 3rd kind.\footnote{Indeed, it holds $\es{R}(\Sigma)=\Span{\smallmat{\I_2\\\0}}$ and $\rho=1$.} As an illustration, define $u_1=\0$, $u_2:t\mapsto [0,0,1,-1]^\intercal$ and $u_3:t\mapsto [1,0,0,0]^\intercal$. Denoting $(u_k,x_k,y_k)\in\fs{Q}(x_0=\0),(k\in\{1,2,3\})$, observe that $y_1=y_2=y_3=\0$. Since inputs $u_k$ are all distinct and $x_1=x_2\neq x_3$, pair $(x_0,y)=(\0,\0)$ is \sing{} of the 3rd kind for $\Sigma$.

    Let $\mathfrak S$ derive from $\Sigma$ by adding following linear constraints
    \begin{align*}
        \mathcal  U & = \left\{[u_a,u_b,u_c,u_d]^\intercal\in\real^4~|~ u_c-u_d=0\right\}, \\ 
        \mathcal  X & = \left\{[x_a,x_b,x_c]^\intercal\in\real^3~|~ x_b+x_c=0\right\}.     
    \end{align*}
    State $x(t)$ belongs identically in $\mathcal{X}$ iff $Ax(t)+Bu(t)\in\es{X} \Leftrightarrow -(x_b(t)+x_c(t)) + u_b(t) + u_c(t) + u_d(t)=0$ for all $t$ and for all $x(t)\in \mathcal{X}$. Together with constraint $u(t)\in \mathcal{U}$, this is equivalent to $u_c=u_d=-u_b/2$. In this case, dynamical equations reduce to
    \begin{align*}
        \dot x_a (t) & = - x_a(t) + u_a(t), \\
        \dot y (t)   & = - y(t) - u_b(t).   \\
    \end{align*}
    Clearly, any pair $(x_0,y=x_c=-x_b)$ uniquely defines $u_b$ and, in turn, $u_c$ and $u_d$. But, it let $u_a$ free. This proves that any admissible pair $(x_0,y)$ is \sing{} of the 2nd kind since distinct $u_a$ leads to distinct $x_a$. This proves that $\mathfrak S$ is IR of the 2nd kind, whereas $\Sigma$ is IR of the 3rd kind.
\end{exmp}

For $k\in\{1,2\}$, the answer to the last question is always negative, i.e. constraints cannot change \singy{} of the 1st or of the 2nd kind. To see this, let $(x_0,y)\in\fs{S}$. Assume that $\Sigma$ is IR of the 1st (resp. 2nd) kind. Then $u_1\neq u_2$ together with $y_1=y_2=y$ implies $x_1=x_2$ (resp. $x_1\neq x_2$) for all $(u_1,x_1,y_1),(u_2,x_2,y_2)\in\fs{Q}_\Sigma(x_0)$. Since $\fs{Q}(x_0)\subseteq \fs{Q}_\Sigma(x_0)$, this proves that $(x_0,y)$ belongs to $\fs{S}_1$ (resp. to $\fs{S}_2$).

\begin{lem}\label{lem:chgkind}
    Let $\Sigma$ be \sing{} of the $k$-th kind, with $k\in\{1,2\}$. If $\mathfrak{S}$ is IR, then $\mathfrak S$ is IR of the $k$-th kind, i.e. $\fs{S}=\fs{S}_k$.
\end{lem}

This lemma proves that if $\Sigma$ is \sing{} of the $k$-th kind, with $k\in\{1,2\}$, then all \sing{} pairs of $\mathfrak S$ (if any) is also of the $k$-th kind. This is in contrast with the case where $k=3$. In this situation, not only $\mathfrak{S}$ can be of the 1st or 2nd kind (see Ex.~\ref{ex:chgkind}), but also \sing{} pairs of different kinds can coexist for $\mathfrak{S}$, as shown by the following example.

\begin{exmp}\label{ex:difkind}
    Let $\Sigma$ be defined as follows:
    \begin{align*}
        \dot x(t) & = \begin{bmatrix} 1 & 0 \\ 0 & 1 \end{bmatrix} x(t) + \begin{bmatrix} 1 & 0 & 0 \\ 0 & 1 & 1 \end{bmatrix} u(t) \\
        y(t)      & = \begin{bmatrix} 0 & 1 \end{bmatrix} x(t).
    \end{align*}
    Constrained system $\mathfrak S$ derives from $\Sigma$ by imposing $\es{U}=[-1;+\infty[^3$ and $\es{X}=[0;1]\times [0;2]$. Let us show that $\mathfrak S$ admits \sing{} pairs of each kind:
    \begin{itemize}
        \item Pair $(x_{0,1},y_1)\coloneqq (\smallmat{1\\0},\0)$ is of the 1st kind. Indeed, first input necessarily equals $t\mapsto -1$ for first state not to escape from $[0;1]$. Together with $y=\0$, this implies that $x$ is constant and equals $t\mapsto x_{0,1}$. Thus, $y_1$ can be produced by any input $u\in\fs{U}_\Sigma$ satisfying $u(t)\in\{-1\}\times (\ker{[1,1]}\cap[-1;+\infty[^2)$ for all $t\in\real_{\geq 0}$.
        \item Pair $(x_{0,2},y_2)\coloneqq (\smallmat{0.5\\2},t\mapsto 2)$ is of the 2nd kind. This time, last two inputs must equal $t\mapsto -1$ for $y(t)=2$ to hold for all $t\in\real_{\geq 0}$. Thus, $y_2$ originates from any input $u\in\fs{U}_\Sigma$ such that (i) $u(t)$ belongs to $[-1;+\infty[\times\{-1\}^2$ for all $t\in\real_{\geq 0}$ and (ii) first state belongs to $[0;1]$. In particular, $\eta_\alpha:t\mapsto -\alpha e^{(1-\alpha)t}/2$ is a suitable trajectory for first input, whatever is $\alpha\in[1;2]$. To see this more easily, note that $\eta_\alpha$ can be interpreted as the feedback of gain $-\alpha$ between first state and first input. Furthermore, any of those input candidates leads to distinct state trajectories. 
        \item Pair $(x_{0,3},y_3)\coloneqq (\smallmat{0.5\\0},\0)$ is of the 3rd kind. This can be proved by constructing triples $(u_4,x_4,y_4)$, $(u_5,x_5,y_5)$, $(u_6,x_6,y_6)\in\fs{Q}(x_{0,3})$ satisfying $y_4=y_5=y_6=y_3$, $x_4=x_5\neq x_6$ and $u_k,(k\in\{4,5,6\})$, all distinct. Using distinct $\alpha_1,\alpha_2\in[1;2]$, it can be verified that $u_4=[\eta_{\alpha_1},\0,\0]^\intercal$, $u_5=(t\mapsto [\eta_{\alpha_1}(t),1,-1]^\intercal)$ and $u_6=[\eta_{\alpha_2},\0,\0]^\intercal$ comply with those constraints.
    \end{itemize}
\end{exmp}

As a last comment, note that constraints can completely destroy redundancy. A trivial two inputs example is the case where $\Sigma$ is IR of the 1st kind, $\ker{\smallmat{B\\D}}=\Span{\smallmat{0 \\ 1}}$, $\es{U}=\{\smallmat{x_a\\x_b}:x_b=0\}$ and $\es{X}=\real^n$.

\section{Linear spaces\label{sec:lincons}}

Let us turn our attention to the case where both $\mathcal U$ and $\mathcal X$ are linear. The following assumption is considered valid throughout this section.

\begin{hypo}[Linear spaces]\label{hyp:linEns}
    Sets $\mathcal U\subseteq \real^m$ and $\mathcal X\subseteq\real^n$ are linear over the field $\real$.
\end{hypo}

By specializing the results of the incremental analysis (see Appendix~\ref{ap:incAn}) supporting results exposed in the previous section, the following theorem can be obtained, see Appendix~\ref{ap:incAn} for the proof. Roughly speaking, this theorem states that ``linearity implies uniformity''.

\begin{thm}\label{th:locglob}
    Under {Hyp.}~\ref{hyp:linEns}, if $\mathfrak S$ admits an \sing{} pair $(x_0,y)\in \mathcal X \times \fs{Y}$ of the $k$-th kind, then $\mathfrak S$ is uniformly IR of the $k$-th kind, i.e. $\fs{A}=\fs{S}=\fs{S}_k\neq\varnothing$.
\end{thm}

As compared with \cite[\locglobUnc]{Kreiss2020} where $\es{U}=\real^m$ and $\es{X}=\real^n$ hold, this theorem is more general since arbitrary linear spaces are considered. It also emphasizes that $\fs{S}$ equals $\fs{S}_k$, i.e. {Hyp.}~\ref{hyp:linEns} prevents pairs of distinct kind to coexist for the same system. This is in stark contrast with the general non linear case, as shown by Ex.~\ref{ex:difkind}.

\ifext{\color{blue}
    This section heavily relies on geometric control theory whose essential aspects are summarized in Appendix~\ref{ap:geom}. Reader is also referred to \cite[Chap.~0]{Wonham1985} or \cite[Sec.~2.2]{Trentelman2001}.
    Following \cite{Trentelman2001}, and relatively to system $\Sigma$, the weakly unobservable subspace and the controllable weakly unobservable subspace are denoted by $\es{V}(\Sigma)$ and $\es{R}(\Sigma)$, respectively.

    \subsection{Description of $\fs{Q}(x_0)$ via trajectories of an unconstrained system}

}\fi

As an initial step, let us get rid of constraints associated with $\es{U}$ and $\es{X}$ by constructing a quadruple whose \emph{unconstrained} trajectories can be gathered into a set isomorphic to $\fs{Q}(x_0)$. \ifext{\color{blue}To this end, this subsection summarizes, formalizes (via Lem.~\ref{lem:Ebij}) and illustrates existing material collected in Appendix~\ref{ap:geomCont}.\newline}\fi
Let $R$ denote the insertion of $\es{U}$ in $\real^m$. Define $B_{\es{U}} \coloneqq  B\mid \mathcal U=BR$ (resp. $D_{\es{U}} \coloneqq  D\mid \mathcal U=DR$) as the domain restriction of $B$ (resp. $D$) to $\es U$. 
Let $\es{V}^*(\es{X},A,B_{\es{U}})$ (or $\es{V}^*(\es{X})$ for short), be the largest $(A,B_{\es{U}})$-controlled invariant subspace contained in $\es{X}$. Pick any friend $F$ in $\underline{F}(\es{V}^*(\es{X}))$, the set of friend of $\es{V}^*(\es{X})$. Select $L$ as any injective 
linear map such that
\begin{equation} \label{eq:Ldef}
    \im L = B_{\es{U}}^{-1} \es{V}^*(\es{X},A,B_{\es{U}}).
\end{equation}
This allows to define \emph{unconstrained} system $\Sigma_F$, characterized by quadruple  $(A_F,B_F,C_F,D_F)$ where $A_F\coloneqq (A+BRF)\mid \es{V}^*(\es{X})$, $B_F\coloneqq \es{V}^*(\es{X})\mid (BR L)$, $C_F\coloneqq (C+DRF)\mid \es{V}^*(\es{X})$ and $D_F\coloneqq DRL$.

\ifext{\color{blue}
    \begin{exmp}[Ex.~\ref{ex:chgkind} continued]\label{ex:SigtoSigF}
        Select following injective matrix $R$ satisfying $\im R=\mathcal  U$:
        \[
            R=\begin{bmatrix}
                1 & 0 & 0   \\
                0 & 1 & 0   \\
                0 & 0 & 1/2 \\
                0 & 0 & 1/2 \\
            \end{bmatrix},
        \]
        so that $B_\mathcal {U}=BR=\I_3$ and $D_\mathcal {U}=DR=\0$ hold. Since $A\mathcal  X$ equals $\mathcal  X$, we trivially have $\mathcal V^*(\mathcal  X)=\mathcal  X$ and $\0\in\underline{F}(\mathcal V^*(\mathcal  X))$. Apply input $u=R(Fx+Lw)$ with $F=\0$ and
        \[
            L=\begin{bmatrix}
                1 & 0  \\
                0 & 1  \\
                0 & -1
            \end{bmatrix},
        \]
        which satisfies \eqref{eq:Ldef}. This leads to $(A_F,B_F,C_F,D_F)=(-\I_2,\I_2,[0,-1],\0)$.
        Indeed, change of variable $\smallmat{\eta \\ \varpi}=P^{-1}x$ with
        \[
            P^{-1} = \begin{bmatrix} 1 & 0 & 0 \\ 0 & 1 & 0 \\ 0 & 1 & 1 \end{bmatrix} \Leftrightarrow P = \begin{bmatrix} 1 & 0 & 0 \\ 0 & 1 & 0 \\ 0 & -1 & 1 \end{bmatrix}
        \]
        gives rise to the following dynamics, with $w$ as new input:
        \begin{align*}\label{eq:SigmaF}
            \begin{bmatrix} \dot\eta \\ \dot\varpi \end{bmatrix} & =P^{-1}(A+B_{\mathcal U}F)P \begin{bmatrix} \eta \\ \varpi \end{bmatrix}+P^{-1}B_{\mathcal U} Lw \\
                                       & = -\I_3 \begin{bmatrix} \eta \\ \varpi \end{bmatrix} + \begin{bmatrix}
                \I_2 \\
                \0
            \end{bmatrix}w
        \end{align*}
        and
        \begin{align*}
            y & = (C+D_{\mathcal U}F)P \begin{bmatrix} \eta \\ \varpi \end{bmatrix} +D_{\mathcal U}Lw \\ & =\begin{bmatrix}
                0 & -1 & 1
            \end{bmatrix}\begin{bmatrix} \eta \\ \varpi \end{bmatrix}.
        \end{align*}
        Since first two columns of $P$ span $\mathcal V^*(\mathcal X)$, $\Sigma_{F=\0}$ is governed by following equation
        \begin{equation}\label{eq:SigmaFex}
            \dot \eta(t) = -\I_2 \eta (t)+ \I_2 w(t),\quad y (t)= \begin{bmatrix} 0 & -1 \end{bmatrix} \eta(t),
        \end{equation}
        corresponding to restriction of the dynamics to~$\es{V}^* (\es{X})$.
    \end{exmp}

    Let $x_0\in\mathcal{X}_c$. By construction, it is possible to link any triple $(u,x,y)\in\fs{Q}(x_0)$ with some element of the set $\fs{Q}_F(\eta_{0})$ which gathers all input $w$, state $\eta$ and output $\varphi$ trajectories originating from initial condition $\eta(0)=\eta_0$ and compatible with $\Sigma_F$: Such a relationship is established via the following mapping:
    \begin{align*} 
        E(F,\eta_0;\cdot):\fs{Q}_F(\eta_0) & \rightarrow  \fs{Q} (T \eta_0)      \\
        (w,\eta,\varphi)                   & \mapsto (RLw+RFT\eta,T\eta,\varphi)
    \end{align*}
    where $T$ denote insertion of $\es{V}^*(\es{X})$ in $\real^n$. This mapping is parametrized by $F$ and $\eta_0\in\real^l$ where
    \begin{equation*}
        l \coloneqq  \dim{\es{V}^*(\es{X},A,B_{\es{U}})}.
    \end{equation*}


    \begin{lem} \label{lem:Ebij}
        Assume that {Hyp.}~\ref{hyp:linEns} and $l>0$ hold. For all $\eta_0\in\real^ l$ and $F\in\underline F(\es{V}^*(\es{X},A,B_{\es{U}}) )$:
        \begin{enumerate}[(i)]
            \item mapping $E(F,\eta_0;\cdot)$ is linear and bijective, so that $\fs{Q}_F(\eta_{0}){\coloneqq E^{-1}(F,\eta_0;\fs{Q}(T\eta_0))}$ and $\fs{Q} (T \eta_0)$ are isomorphic;
            \item {for all $(w,\eta,\varphi)\in\fs{Q}_F(\eta_0)$, input $w$ is causal, piecewise continuous and exponentially bounded.}
        \end{enumerate}
    \end{lem}

    \begin{proof}
        $E(F,\eta_0;\cdot)$ is surjective by construction. Its injectivity comes from that of $RL$. Linearity is obvious. {This proves (i). As far as (ii) is concerned, pick any $(u,x,y)\coloneqq (RLw+RFT\eta,T\eta,\varphi)$ in $\fs Q(T\eta_0)\subseteq \fs U_\Sigma \times \fs{C}(\real_{\geq 0},\real^n)\times \fs Y$. Since $u\in\fs{U}_\Sigma$ holds, $x=\H_x[T\eta_0;u]$ is causal, continuous and exponentially bounded. Matrix $T$ being injective, $\eta$ enjoys the same property. Due to equality $RLw=u-RFT\eta$ and injectivity of $RL$, $w$ inherits (a) piecewise continuity from $u$ and (b) exponential boundedness and causality from $u$ and~$\eta$.}
    \end{proof}

    Let us emphasize that any triple $(w,\eta,\varphi)$ satisfying differential-algebraic equations analogous to \eqref{eq:sys} (but associated with $\Sigma_F$) belongs to $\fs{Q}_F(x_0)$. This is in stark contrast with triple $(u,x,y)\in\fs{Q}(x_0)$ which satisfies \eqref{eq:sys} and, at the same time, belongs to $\fs{U} \times \fs{X} \times \fs{Y}$ (see \eqref{eq:Qdef}). Notwithstanding, $E(F,\eta_0;\cdot)$ maps $\fs{Q}_F(x_0)$ to $\fs{Q}(x_0)$ in a bijective way. To arrive at this result, input and state constraints associated to $\mathfrak S$ have been somehow \emph{structurally embed} into quadruple of $\Sigma_F$.

    \begin{exmp}[Ex.~\ref{ex:chgkind} continued]
        Exerting input $w:t\mapsto [1,1]^\intercal$ to $\Sigma_{F=\0}$ for $\eta_0=[1,0]^\intercal$ produces $\eta:t\mapsto[1,1-e^{-t}]$ and, in turn, $y=\varphi:t\mapsto e^{-t}-1$.
        Since $T$ can be selected as the first two columns of $P$, one can check that
        \begin{multline*}
            (u,x,y) =  E(\0,\eta_0;(w,\eta,\varphi)) \\
            =(t\mapsto \smallmat{ 1 \\ 1\\ -1/2\\-1/2 }, t\mapsto \smallmat{  1 \\ 1 -e^{-t}\\ e^{-t}-1},t\mapsto e^{-t}-1)
        \end{multline*}
        belongs to $\fs{Q}(x_0=[ 1,0,0]^\intercal)$, i.e. both \eqref{eq:sys} and $(u(t),x(t))\in\mathcal U \times\mathcal X$ hold for all $t\in\real_ {\geq 0}$.
    \end{exmp}

    \subsection{Characterizations of IR}

    In this subsection, it is shown how the results of \cite{Kreiss2020} can be extrapolated to the context of current Section~\ref{sec:lincons}.

    Assume that $l>0$ hold, and select any $\eta_0\in\real^ l$ and $F\in\underline F(\es{V}^*(\es{X},A,B_{\es{U}}) )$. Let $x_0\coloneqq T\eta_0$. For $k\in\{1,2\}$, pick any $(w_k,\eta_k,\varphi_k)\in\fs{Q}_F(\eta_{0})$ such that $\varphi_1=\varphi_2$ and $w_1\neq w_2$ hold. Define $(u_k,x_k,y_k)\coloneqq E(F,\eta_0;(w_k,\eta_k,\varphi_k))$. Then, (i) $y_1=y_2$ holds, by definition of $E(F,\eta_0;\cdot)$, and (ii) $u_1\neq u_2$ holds, since $u_1=u_2\Rightarrow x_1=x_2$ which would contradict injectivity $E(F,\eta_0;\cdot)$. This proves that if $(\eta_0,\varphi_1)$ is an IR pair for $\Sigma_F$, then $(x_0,y_1)$ is an IR pair for $\mathfrak S$. Using the same reasoning, one can prove that the opposite implication is also valid, i.e. if $(x_0,y_1)$ is an IR pair for $\mathfrak S$, then $(\eta_0,\varphi_1)$ is an IR pair for $\Sigma_F$.
    Besides, equality $\eta_1= \eta_2$ is equivalent to $x_1= x_2$, by definition of $E(F,\eta_0;\cdot)$ and from injectivity of $T$. This allows to conclude that $(\eta_0,\varphi_1)$ is an IR pair for $\Sigma_F$ of the $k$-th kind iff $(x_0,y_1)$ is an IR pair for $\mathfrak S$ of the same kind.

}\fi

\ifext{\color{blue}
    From this discussion and Th.\ref{th:locglob}, one can readily extrapolate the results of \cite{Kreiss2020} since Lem.~\ref{lem:Ebij} ensures that $w_k$ are causal, piecewise continuous and exponentially bounded. For the sake of completeness, the conclusions drawn in this way are now exposed.
}\fi

\begin{thm}\label{th:charac}
    Assume that
    \ifext{\color{blue}$l>0$}\fi and {Hyp.}~\ref{hyp:linEns} hold. Select any friend $F$ of $\es{V}^*(\es{X},A,B_{\es{U}})$. Define integers $\rho_F$, $\nu_F$ and subspace $\es{N}_F\subseteq \es{U}$ as follows:
    \begin{align}\label{eq:rhodef}
        \rho_F   & \coloneqq \dim (\ker{B_F}\cap\ker{D_F}),                \\
        \es{N}_F & \coloneqq B_F^{-1}\mathcal{V}(\Sigma_F) \cap \ker{D_F}, \\ \label{eq:mudef}
        \nu_F    & \coloneqq \dim ( \es{N}_F / \ker{\smallmat{B_F          \\ D_F}} ) = \dim(\es{N}_F) -\rho_F.
    \end{align}
    Then, the following statements are equivalent:
    \begin{enumerate}[(i)]
        \item System $\mathfrak S$ is IR;
        \item $\nu_F>0$ or $\rho_F>0$;
        \item Transfer matrix $G_F(s)\coloneqq C_F (s\I -A_F)^{-1}B_F+D_F$ of $\Sigma_F$ is not left-invertible,
        \item System matrix $P_F(s)\coloneqq \smallmat{ s\I -A_F & -B_F \\ C_F & D_F }$ of $\Sigma_F$ is not left-invertible. 
    \end{enumerate}
    Furthermore, the kind of redundancy of $\mathfrak S$ is characterized by $\rho_F$ and $\nu_F$, as in the following table:
    \begin{center}
        \begin{tabular}{|c|c|c|} \hline
            IR       & $\rho_F$ & $\nu_F$ \\ \hline \hline
            1st kind & $>0$     & $=0$    \\ \hline
            2nd kind & $=0$     & $>0$    \\ \hline
            3rd kind & $>0$     & $>0$    \\ \hline
        \end{tabular}
    \end{center}
    Besides, it holds
    \begin{equation*}\label{eq:nudim}
        \dim\es{R}(\Sigma_F) >0 \;\Leftrightarrow\; \nu_F>0.
    \end{equation*}
\end{thm}

\ifext{\color{blue}
    \begin{exmp}[Ex.~\ref{ex:chgkind} continued]\label{ex:SigtoSigFdeg}
        For this example, one gets $\es{R}(\Sigma_F)=\mathcal{V}(\Sigma_F)=\es{N}_F=\im{\smallmat{1 \\ 0}}$ and $\rho_F=0$. Thus, $\mathfrak S$ is IR of the 2nd kind with degree $(0,1)$.

        Let us exemplify discussion conducted in \cite[\matview]{Kreiss2020} on $\Sigma_F$. First note that original input and state basis of $\Sigma_F$ are already adapted to $\es{N}_F$ and $\es{R}(\Sigma_F)$, respectively. Thus, $\Sigma_F$ enjoys a cascaded structure without regular feedback. Here, this cascade actually degenerates into two decoupled subsystems. Indeed, \eqref{eq:SigmaFex} can be rewritten as follows:
        \begin{equation*}
            \dot \eta_a (t)= -\eta_a (t)+ w_a(t),\quad \dot y(t) = -y(t) - w_b (t),
        \end{equation*}
        using notations $w\coloneqq [w_a,w_b]^\intercal$ and $\eta\coloneqq [\eta_a,\eta_b]^\intercal$.
        The crucial point is the following: If $w_b$ is uniquely defined by initial condition $x(0)=[\eta^\intercal (0),y(0)]^\intercal$ and output $y$, first input $w_a$ does not impact $y$ and, hence, can be arbitrarily selected. For instance, pick $w_1:t\mapsto [1,0]$ and $w_2=\0$. Those inputs produce $\eta_1:t\mapsto [1-e^{-t},0]$ and $\eta_2=\0$ from zero initial condition, respectively. By $E$, one derives corresponding triples: It holds $(u_1,x_1,y_1)=(t\mapsto [1,0,0,0]^\intercal,t\mapsto [1-e^{-t},0,0],\0)$ and $(u_2,x_2,y_2)=(\0,\0,\0)$, so that both admissible $u_1$ and $u_2$ produce admissible state trajectories and identical outputs.
    \end{exmp}
}\fi

To conclude, note that the concept of degree of IR defined in \cite[\secdeg]{Kreiss2020} can be readily generalized in this context.

\begin{defn}[Degree of IR]
    Assume that {Hyp.}~\ref{hyp:linEns} and $l>0$ hold.
    Pair $(\rho_F,\nu_F)$, defined via \eqref{eq:rhodef} and \eqref{eq:mudef}, is the \emph{degree of redundancy} of system $\mathfrak S$.
\end{defn}

\ifext{\color{blue}

    \subsection{Additional remarks}

    \begin{rem}[If $l=0$] \label{rem:l0}
        Consider the degenerate case where $\es{V}^*(\es{X},A,B_{\es{U}})$ equals $\{\0\}$, i.e. $l=0$. In such a situation, state trajectory must start from the origin and cannot escape from this point without violating input or state constraints. System $\mathfrak S$ is somehow over-constrained and state-space of $\Sigma_F$ collapses to $\{\0\}$. Therefore, set $\fs{U}_a(x_0)$ of admissible inputs for $x_0$ is non empty iff $x_0=\0$. In this case, $\fs{U}_a(x_0)$ reads $\fs{U}_\Sigma\cap \{u:\real_ {\geq 0}\rightarrow\es{U}\cap\ker B\}$. Furthermore, $\Sigma_F$ degenerates into static input-output map $w\mapsto \varphi = Mw$ where $M\coloneqq D\mid (\es{U}\cap\ker B)$. It should be clear that $\mathfrak S$ is IR iff $\dim{\es{U}\cap\ker{\smallmat{B\\ D}}}=\dim\ker M >0$. In such a case, every \sing{} pair is of the 1st kind.
    \end{rem}
}\fi

\begin{rem}[Dependency w.r.t. $F$]
    From discussion above, 
    one concludes that $T\es{V}(\Sigma_F)$ is the weakly unobservable subspace of quadruple $(A,B_{\es{U}},\smallmat{C \\ X^\perp},\smallmat{D_{\es{U}} \\ \0})$, where $X^\perp$ satisfies $\ker{X^\perp}=\es{X}$. It follows that $\es{V}(\Sigma_F)$ is independent of $F\in\es{V}^*(\es{X})$. Observe that $B_F$, $D_F$ and, in turn, $\rho_F$ do not depend on $F$ either.\footnote{Subscript $F$ of $B_F$ and $D_F$ aims distinguishing those matrices from $B$ and $D$.}
    As a result, $\es{N}_F$ and $\nu_F$ too are independent on the selection of $F$.
    In the same vein, note that friend $F\in\underline F(\es{V}^*(\es{X}) )$ can be arbitrarily selected in Th.~\ref{th:charac}. This proves that left-invertibility of $G_F(s)$ and $P_F(s)$ are independent of $F$.
\end{rem}

\section{Conclusions}

This paper investigates how input and state constraints affect IR. It is perhaps the most natural extension of results obtained in \cite{Kreiss2020}, and the most desirable from the control application point of view.

The general case where $\es{U}$ and $\es{X}$ are arbitrary is first investigated. If those sets are closed, then Cor.~\ref{cor:notFenetre} proves that admissible pairs $(x_0,y)$ losing \singy{} due to the constraints are necessarily associated with trajectories that lie on the boundary of $\es{U}$ and $\es{X}$. It is also shown that constraints might (i) change the kind of \singy{} of the system (see Ex.~\ref{ex:chgkind}) and (ii) give rise to \sing{} pairs of different kinds coexisting for the same system (see Ex.~\ref{ex:difkind}). Lem.~\ref{lem:chgkind} also proves that such a phenomena can be observed only if $\Sigma$ is IR of the 3rd kind. Otherwise, the kind of redundancy is preserved, provided that redundancy itself is not destroyed by the constraints.

Whenever both input and state spaces are linear, admissible trajectories of $\mathfrak{S}$ are uniquely associated to that of an \emph{unconstrained} auxiliary system $\Sigma_F$, by way of some bijective mapping. This key feature allows to extrapolate results derived in \cite{Kreiss2020}.

In the general case, condition for admissible $(x_0,y)$ to be \sing{} is only sufficient (see Th.~\ref{th:singNL}). Refinement of this condition as well as derivation of a necessary counterpart might improve understanding of how geometry of constraints impacts IR.

Given an \sing{} pair $(x_0,y)$, so that there exist $u_{1}\neq u_{2}$ satisfying $(u_1,y),(u_2,y)\in\fs W(x_0)$. In the linear case, it is implicitly proved that $u_{1,2}$ can be selected as ``different'' as desired. In practice this feature is of major interest since it is expected that ``insufficiently different'' inputs might be useless for redundancy to be exploited. Since this property is destroyed in the non linear case, a criterion ensuring that $(x_0,y)$ is ``\sing{} enough'' would be desirable.

Another possible extension of this work is the refinement of results of Sect.~\ref{sec:NL} in the case where $\es{U}$ and $\es{X}$ enjoy specific properties like e.g. convexity. In this case, it is expected that tighter approximation of $\fs{S}$ can be derived.

\bibliographystyle{ieeetr}
\bibliography{ref}           

\appendix

\ifext{\color{blue}
\section{Background on geometric control theory\label{ap:geom}}

Let us summarized some essential aspects of geometric control theory for system $\Sigma$ governed by \eqref{eq:sys}.

\subsection{Controlled invariance}

A subspace $\mathcal W$ is said to be an $(A , B)$-\emph{controlled invariant} subspace (or simply \emph{controlled invariant} subspace) if, for any initial state $x_0\in\mathcal W$, there exists an input function $u$ such that the state trajectory generated by the system remains identically in $\mathcal W$. Subspace $\mathcal W$ enjoys this property iff there exists a matrix $F$ such that
\begin{equation}\label{eq:charFriend}
    (A+BF)\mathcal W\subseteq \mathcal W
\end{equation}
holds. Matrix $F$ is called a \emph{friend} of $\mathcal W$ and the set of such matrices is denoted by $\underline F(\mathcal W)\coloneqq \{F:\eqref{eq:charFriend}\}$.

Consider any $x_0\in\mathcal{W}$, $F\in\underline F(\mathcal W)$ and $L$ such that $\im L\in  B^{-1}\mathcal{W}$. Then, state trajectory remains in $\mathcal W$ iff there exists a function $w$ such that corresponding input $u$ reads
\begin{equation}\label{eq:ukappa}
    u:t\mapsto \kappa_w (t,x(t)),
\end{equation}
where time-varying feedback $\kappa_w$ is defined as follows
\begin{equation} \label{eq:kappadef}
    \kappa_w:(t,x)\mapsto Fx+Lw(t),
\end{equation}
and is parametrized by $w$ \cite[Th. 4.3]{Trentelman2001}. In particular, selecting $w=\0$ proves that simple feedback $u(t)=Fx(t)$ makes $\mathcal W$ controlled invariant. This input function admits implicit form $u=F\H_x[x_0;u]$ from which  closed-form expression $u:t\mapsto Fe^{t(A+BF)}x_0$ can be derived by means of \eqref{eq:sysAB}.

\subsection{Weak unobservability}

Following \cite[p.159]{Trentelman2001}, a point $x_0\in\real^n$ is called \emph{weakly unobservable} if there exists an input function $u$ such that corresponding output is identically null. The set of all weakly unobservable point is denoted by $\mathcal{V}(\Sigma)$, i.e.
\begin{equation*}
    \mathcal{V}(\Sigma) \coloneqq  \{x_0 \in \real^n \mid \exists u:\real_{\geq 0}\rightarrow \real^m : \H[x_0;u]=\N\}.
\end{equation*}
It can be proved that $\mathcal{V}(\Sigma)$ is a linear subspace of $\real^n$. Referring to the following inclusion parametrized by subspace $\mathcal W\subseteq\real^n$,
\begin{equation} \label{eq:charFriendV}
    \mathcal{W} \subseteq \ker{C+DF},
\end{equation}
set $\mathcal{V}(\Sigma)$ is the largest subspace $\mathcal W$ for which there exists $F$ such that \eqref{eq:charFriend} and \eqref{eq:charFriendV} hold \cite[Th.~7.10]{Trentelman2001}. Such a matrix $F$ is called a friend of $\mathcal V(\Sigma)$ and the set of such friends is denoted by $\underline F(\mathcal V(\Sigma))$.\footnote{As in \cite{Trentelman2001}, two definitions of $\underline F(\mathcal W)$ coexist depending if $\mathcal W$ is $\mathcal V(\Sigma)$ or any other subspace. The context should clarify which one is referred to, though.}

Consider any $x_0\in\mathcal{V}(\Sigma)$, $F\in\underline F(\mathcal V(\Sigma))$ and $L$ satisfying
\begin{equation} \label{eq:Limdef}
    \im L\in B^{-1}\mathcal V(\Sigma) \cap \ker D.
\end{equation}
Then, input $u:\real_{\geq 0}\rightarrow \real^m$ satisfies $\H[x_0;u]=\N$ iff it reads \eqref{eq:ukappa} where $\kappa_w$ is given by \eqref{eq:kappadef} and is parametrized by some function $w$ \cite[Th.~7.11]{Trentelman2001}. Selecting $w=\0$ yields $u(t)=Fx(t)$ so that output $y$ is governed by $\dot x = (A+BF)x,\;y=(C+DF)x$ (see \eqref{eq:sys}). From previous characterization of $\mathcal V(\Sigma)$ via \eqref{eq:charFriend} and \eqref{eq:charFriendV}, this input not only makes output identically zero ($\mathcal{V}(\Sigma) \subseteq \ker{C+DF}$) but also enforces state trajectory to remain in $\mathcal V(\Sigma)$  ($(A+BF)\mathcal{V}(\Sigma) \subseteq \mathcal{V}(\Sigma)$). This explains alternative naming of $\mathcal V(\Sigma)$ as the largest output nulling controlled invariant subspace \cite{Anderson1975}.

Note that if $D$ equals $\0$, then $\mathcal{V}(\Sigma)$ reduces to $\es{V}^*(\ker{C},A,B)$, the largest $(A,B)$-controlled invariant subspace contained in the kernel of $C$.


\subsection{Controllability and weak unobservability}

From \cite[p.163]{Trentelman2001}, a point $x_0\in\real^n$ is called \emph{controllable weakly unobservable} if there exist $T>0$ and $u:[0,T]\rightarrow \real^m$ such that the following conditions hold with $x_f=\0$:
\begin{subequations} \label{eq:Rstarcond}
    \begin{align}\label{eq:Rstarcondy}
        \H[x_0;u](t)   & =\0,\; \forall t\in[0,T], \\ \label{eq:Rstarcondx}
        \H_x[x_0;u](T) & =x_f.
    \end{align}
\end{subequations}
The set of all controllable weakly unobservable points is denoted by $\es{R}(\Sigma)$. As suggested by its definition, $\es{R}(\Sigma)$ is included in $\mathcal{V}(\Sigma)$. It can also be proved that $\es{R}(\Sigma)$ is a linear subspace of $\real^n$ which satisfies \eqref{eq:charFriend} and \eqref{eq:charFriendV} with $\mathcal W=\es{R}(\Sigma)$ and any $F$ in $\underline F(\mathcal V(\Sigma))$.

\subsection{A matrix view point}

Pick any $F\in\underline F(\mathcal V(\Sigma))$ and $L$ satisfying \eqref{eq:Limdef}. Choose a basis of state space $\real^n$ adapted to $\es{R}(\Sigma)$ and $\mathcal{V}(\Sigma)$, i.e. define invertible matrices $T=[T_a,T_b,T_c]$ satisfying $\im{T_a}=\es{R}(\Sigma)$ and $\im{[T_a,T_b]}=\mathcal{V}(\Sigma)$. Apply input $u$ given  by \eqref{eq:ukappa} with \eqref{eq:kappadef} to $\Sigma$. Resulting closed-loop system is such that new input $w$ drives output $y$ via quadruple $(A+BF,BL,C+DF,DL)$. From the above discussion, dynamics in the new basis reads:
\begin{subequations} \label{eq:decIntExt}
    \begin{align}
        \begin{bmatrix} \dot \phi_a \\ \dot \phi_b \\ \dot \phi_c \end{bmatrix} & =
        \begin{bmatrix} A_{11} &  A_{12} &  A_{13} \\ \0 & A_{22} &  A_{23} \\ \0 & \0 &  A_{33}  \end{bmatrix} \begin{bmatrix}  \phi_a \\  \phi_b \\  \phi_c \end{bmatrix} + \begin{bmatrix} B_{1} \\ \0 \\ \0 \end{bmatrix} w, \\
        y                          & = \begin{bmatrix} \0 & \0 & C_{3} \end{bmatrix} \begin{bmatrix}  \phi_a \\  \phi_b \\  \phi_c \end{bmatrix}.
    \end{align}
\end{subequations}
Here, $\phi_i$ denotes coordinates of $x$ in $\im{T_i}$, so that $[\phi_a^\intercal,\phi_b^\intercal,\phi_c^\intercal]^\intercal=T^{-1}x$ is the new state vector. One recovers that for any $w$ and $x_0\in\es{V}(\Sigma)$ (so that $\phi_c(0)=\0$), state trajectory remains identically in $\es{V}(\Sigma)$ (since substate $\phi_c$ equals $\0$), yielding $y=\0$.

\subsection{Facts about controllability\label{sec:contr}}

From \eqref{eq:Rstarcond}, $\es{R}(\Sigma)$ gathers initial states which are output-nulling ``controllable-to-the-origin''. It actually coincides with the set of output-nulling reachable states (reachability being ``controllability-from-the-origin'') \cite[p.170]{Trentelman2001}, i.e. $x_f$ belongs to $\es{R}(\Sigma)$ iff there exist $T>0$ and $u:[0,T]\rightarrow \real^m$ such that \eqref{eq:Rstarcondy} and \eqref{eq:Rstarcondx} hold with $x_0=\0$. Among other things, this proves that pair $( A_{11}, B_1)$ is controllable.
To sees this, let $x_0=\0$, pick any $\phi_{a,f}$ and define $x_f\coloneqq T_a\phi_{1,f}\in\es{R}(\Sigma)$. Since \eqref{eq:Rstarcondy} holds for some $T$ and $u$, so that the latter reads \eqref{eq:ukappa} with \eqref{eq:kappadef} and resulting dynamics is governed by \eqref{eq:decIntExt}. In the new coordinates, \eqref{eq:Rstarcondx} reads $\phi_{a}(T)=\phi_{a,f}$ which proves controllability of $( A_{11}, B_1)$  since $T$ and $\phi_{a,f}$ are arbitrary.

Standard result on controllability (see e.g. \cite[Cor. 2.14]{Antsaklis2006}) ensures that substate $\phi_a$ can be transferred from any $\phi_{a,0}$ to any $\phi_{a,f}$ in arbitrary finite time $T>0$ by way of input defined in \eqref{eq:wtrans}, i.e. $\phi_{a}(T)=\phi_{a,f}$ when $\phi_{a}(0)=\phi_{a,0}$ and $w$ reads as follows
\begin{equation} \label{eq:wtrans}
    w:t\mapsto B_1^\intercal e^{A_{11}^\intercal (T-t)} W_r^{-1}(T) (\phi_{a,f}-e^{A_{11}T}\phi_{a,0}),
\end{equation}
where $W_r:T\mapsto \textstyle\int_0^T e^{(T-\tau)A_{11}}B_1B_1^\intercal e^{(T-\tau)A_{11}^\intercal}\text{d}\tau$ is the reachability Gramian from $0$ to $T$. Translated in the original coordinates, this discussion ensures that $x$ can be transferred from any $x_0\in\es{R}(\Sigma)$ to any $x_f\in\es{R}(\Sigma)$ in arbitrary finite time $T>0$ and in such a way that both $x(t)\in\es{R}(\Sigma)$ and $y(t)=\0$ hold for all $t\in[0,T]$. Input $u:t\mapsto F T_a \phi_a (t)+Lw(t)$ achieves such a result. Here, $F\in\underline F(\mathcal V(\Sigma))$ holds, $\phi_a$ satisfies $\dot \phi_a = A_{11} \phi_a+B_1w$ and $w$ is given by \eqref{eq:wtrans} where $\phi_{a,f}$ and $\phi_{a,0}$ are such $x_f=T_a\phi_{1,f}$ and $x_0=T_a \phi_{a,0}$ hold. From this expression, it comes out that $u$ is continuous on $[0,T]$. Indeed, $\phi_a$ inherits continuity from $w$ and, together, they prove that $u$ enjoys this property as well.

\subsection{Linear input and state constraints \label{ap:geomCont}}

Consider now the case where following constraints is imposed on the system for all $t\in\real_{\geq 0}$ :
\begin{equation*}
    (u(t),x(t))\in\es{U}\times\es{X},
\end{equation*}
where $\es U \subseteq \real^m$ and $\es X\subseteq \real^n$ are given \emph{linear} sets.

Let $R$ denote the insertion of $\es{U}$ in $\real^m$. Define
$
    B_{\es{U}} \coloneqq  B\mid \mathcal U=BR,
$
the domain restriction of $B$ to $\es U$. 
Pick any $x_0\in\es{X}_c$ and any $(u,x,y)$ in $\fs{Q}(x_0)$. In this case, state trajectory $x$ lies in the largest set contained in $\es{X}$ which can be made invariant by way of some input $u\in\fs{U}$. This set inherits linearity property of $\es{X}$ and corresponds to $\es{V}^*(\es{X},A,B_{\es{U}})$ (or $\es{V}^*(\es{X})$ for short), the largest $(A,B_{\es{U}})$-controlled invariant subspace contained in $\es{X}$.
The set of input $u$ producing such a dynamics can be parametrized by way of any matrix $F$ in $\underline F(\es{V}^*(\es{X}))$, the set of friend of $\es{V}^*(\es{X})$.
Indeed, in this context, \cite[Th. 4.3 and Th. 4.5]{Trentelman2001} apply and ensures that $u$ produces $x$ in $\fs{X}$ iff $x(0)$ belongs to $\es{V}^*(\es{X})$ and $u$ equals $R(Fx+Lw)$ for some signal $w$ of appropriate dimension. Here, $L$ is an injective 
linear map such that
\begin{equation}
    \im L = B_{\es{U}}^{-1} \es{V}^*(\es{X},A,B_{\es{U}}).
\end{equation}
Using this expression of $u$ transforms quadruple $\Sigma$ into $(A+B_{\es{U}} F , B_{\es{U}} L , C+D_{\es{U}}F , D_{\es{U}}L)$ where $D_{\es{U}} \coloneqq  D\mid \mathcal U=DR$ is the domain restriction of $D$ to $\es U$.
Let $\Sigma_F$ be the restriction of dynamics induced by this new quadruple to $\es{V}^*(\es{X})$, i.e. quadruple  $(A_F,B_F,C_F,D_F)$ of $\Sigma_F$ satisfies $A_F\coloneqq (A+B_{\es{U}}F)\mid \es{V}^*(\es{X})$, $B_F\coloneqq \es{V}^*(\es{X})\mid (B_{\es{U}} L)$, $C_F\coloneqq (C+D_{\es{U}}F)\mid \es{V}^*(\es{X})$ and $D_F\coloneqq D_{\es{U}}L$.

\begin{exmp}[When $\es{V}^*(\es{X})$ is a proper subet of $\mathcal X$]
    Let us clarify that $\es{V}^*(\es{X})$ can be strictly included in $\mathcal X$, so that $\dim{\es{V}^*(\es{X})}<\dim{\es{X}}$. As an example, consider following dynamical system
    \begin{align*}
        \dot x = \begin{bmatrix} 0 & -1 &  0 \\ 1 & 0 & 0 \\ 0 & 0 & 1 \end{bmatrix} x + \begin{bmatrix} 0 \\ 0 \\ 1 \end{bmatrix}u
    \end{align*}
    associated with constraints $\es{U}=\real$ and $\es{X}=\{0\}\oplus\real^2$, translating the fact that first state is enforced to be zero. In this case, input cannot modify first two states, whose dynamics are coupled. As a result, second state must be identically null and $\es{V}^*(\es{X})$ reduces to $\{0\}^2\oplus\real$.
\end{exmp}

Let $T$ denote insertion of $\es{V}^*(\es{X})$ in $\real^n$. We have just proved that (i) $\fs{Q}(x_0)$ is non empty iff $x_0$ belongs to $\es{V}^*(\es{X})$, i.e. $x_0=T\eta_0$ for some $\eta_0$, and (ii) $(u,x,y)$ belongs to $\fs{Q} (T\eta_0)$ iff there exists $(w,\eta,\varphi)\in\fs{Q}_F(\eta_{0})$. Here, set $\fs{Q}_F(\eta_{0})$ gathers all input $w$, state $\eta$ and output $\varphi$ trajectories originating from initial condition $\eta(0)=\eta_0$ and compatible with system $\Sigma_F$. More precisely, if $l>0$, then the following mapping is surjective for all $\underline F(\es{V}^*(\es{X}))$ and $\eta_0$:
\begin{align*} 
    E(F,\eta_0;\cdot):\fs{Q}_F(\eta_0) & \rightarrow  \fs{Q} (T \eta_0)      \\
    (w,\eta,\varphi)                   & \mapsto (RLw+RFT\eta,T\eta,\varphi)
\end{align*}

\begin{rem}[About $\es{X}_c$ and $\fs{U}_a(x_0)$]
    Note that $\es{X}_c$ is nothing but $\es{V}^*(\es{X})$. Furthermore, if $l>0$ and $x_0\in\es{X}_c\coloneqq  \{x_0\in\real^n: \card{\fs{W} (x_0) }\geq 1\} \subseteq \mathcal X$, then $\fs{U}_a(x_0) \coloneqq  \{u\in\fs{U} \mid \exists y:(u,y)\in\fs{W}(x_0)\}$ reads
    \begin{multline*}
        \{RLw+RFT\eta \mid \\
        RLw\in\fs{U}_\Sigma,\;\dot\eta=A_F\eta+B_Fw,\;T\eta(0)=x_0\},
    \end{multline*}
    for any $F\in\underline F(\es{V}^*(\es{X}) )$.
\end{rem}

}\fi

\section{Incremental analysis\label{ap:incAn}}

This appendix aims analyzing input-to-output mapping associated with $\mathfrak{S}$ by adopting an incremental view point: Inputs pair $(u_1,u_2)$ are described via $(u_1,\tilde u)$ where $\tilde u=u_2-u_1$ so that $u_2=u_1+\tilde u$ can be expressed \emph{relatively} to $u_1$.

\ifext{\color{blue} Instrumental are the following well-known relationships which hold for all $x_0,\tilde x_0\in \real^n$, $u,\tilde u\in\fs{U}_\Sigma$ and $\alpha\in\real$, due to linearity of \eqref{eq:sys}:
    \begin{subequations}
        \begin{align}\label{eq:lin1}
            \H_x[x_{0};u]+\H_x[\tilde x_{0};\tilde u] & = \H_x[x_{0}+\tilde x_{0};u+\tilde u], \\ \label{eq:lin2}
            \H_x[\alpha x_{0};\alpha u]               & = \alpha  \H_x[x_{0};u].
        \end{align}
    \end{subequations}
    The same two properties also hold for $\H$ in place of~$\H_x$.
}\fi

\subsection{Incremental characterization of \singy{}}



Given a pair $(x_0,y)\in\es{X}\times\fs{Y}$.
\ifext{\color{blue}
    Recall that set $\mathfrak{H}^{-1}[x_0;y]$, defined in \eqref{eq:Hinv}, gathers all admissible input $u$ giving rise to output $y$ when state is initialized at $x_0$.}
\fi
By definition, $\mathfrak{H}^{-1}[x_0;y]$ is non empty iff $(x_0,y)\in\fs{A}$ holds.

Pair $(x_0,y)$ is \sing{} iff $\card{\mathfrak{H}^{-1}[x_0;y]}\geq 2$ holds.
In order to derive conditions under which this inequality holds, let us describe $\mathfrak{H}^{-1}[x_0;y]$ \emph{relatively} to some element $u$ of this set. Specifically, given any $u\in\mathfrak{H}^{-1}[x_0;y]$, define $\widetilde{\fs{U}}(x_0,u)$ as follows:
\begin{equation} \label{eq:eqCDinc}
    \widetilde{\fs{U}}(x_0,u) \coloneqq  \mathfrak{H}^{-1}[x_0;y] - u.
\end{equation}
Note that $\widetilde{\fs{U}}$ is parametrized by $(x_0,u)\in \mathcal X \times \fs{U}$, but not by $y$ which can be uniquely\footnote{Recall that signals in $\fs{U}\subseteq \fs{U}_\Sigma$ are piecewise continuous.} recovered from $x_0$ and $u$. Putting \eqref{eq:eqCDinc} in the other way around makes it easier to grasp: Equality $\mathfrak{H}^{-1}[x_0;y] = \widetilde{\fs{U}}(x_0,u)+ u$ simply says that $\widetilde{\fs{U}}(x_0,u)$ gathers signals $\tilde u:\real_{\geq 0}\rightarrow \real^m$  such that $u+\tilde u$ belongs to $\mathfrak{H}^{-1}[x_0;y]$ or, equivalently, satisfies $(u+\tilde u,y)\in\fs{W}(x_0)$. In this case, both inputs $u$ and $u+\tilde u$ comply with the constraints and produce identical output $y$ for the same initial state $x_0$.

Observing that $\N$ trivially belongs to $\widetilde{\fs{U}}(x_0,u)$, this discussion leads to a new characterization of \singy{}.

\begin{prop} \label{prop:singD}
    Pair $(x_0,y)\in \fs{A}$ is \sing{} iff the following condition holds
    \begin{equation} \label{eq:singD}
        \forall u\in\mathfrak{H}^{-1}[x_0;y]: \card{\widetilde{\fs{U}}(x_0,u)} \geq 2.
    \end{equation}
    Further, set $\widetilde{\fs{U}}(x_0,u)$ reads
    \begin{multline} \label{eq:Ddef}
        \widetilde{\fs{U}}(x_0,u)=\{\tilde u:\real_{\geq 0}\rightarrow \real^m :\\
        (\tilde u,\N)\in \fs{W} (x_0)-(u,\H[x_0;u]) \}.
    \end{multline}
\end{prop}

\begin{proof}
    Pick any $u$ in non empty set $\mathfrak{H}^{-1}[x_0;y]$ parametrized by $(x_0,y)\in \fs{A}$. Equality between $\card{\mathfrak{H}^{-1}[x_0;y]}$ and $\card{\widetilde{\fs{U}}(x_0,u)}$ follows from \eqref{eq:eqCDinc}, so that \singy{} of $(x_0,y)$ is equivalent to $\card{\widetilde{\fs{U}}(x_0,u)} \geq 2$. Input $u$ being  arbitrary, this proves~\eqref{eq:singD}.


    To prove \eqref{eq:Ddef}, 	observe that $(u+\tilde u,y)=(u,y)+(\tilde u,\N)$ belongs to $\fs{W}(x_0)$ for all $\tilde u\in\widetilde{\fs{U}}(x_0,u)$, which is equivalent to saying that $\tilde u$ satisfies $(\tilde u,\N)\in\fs{W}(x_0)-(u,y)$. Since $y=\H[x_0;u]$, this proves \eqref{eq:Ddef}.
\end{proof}

Consider the following technical result.

\begin{lem}\label{lem:Qtildeexpr}
    Given arbitrary $x_{0,1},x_{0,2}\in\es X_c$ and $(u,x,y)\in\fs Q(x_{0,2})$, set $\fs{Q}(x_{0,1})-(u,x,y)$ reads
    \begin{multline} \label{eq:Qtildeexpr}
        \{(\tilde{u},\tilde x, \tilde{y})\in\fs{U}\times\fs{X}\times\fs{Y}-(u,x,y):\\
        \H_x[x_{0,1}-x_{0,2};\tilde{u}]=\tilde x,\H[x_{0,1}-x_{0,2};\tilde{u}]=\tilde y\}.
    \end{multline}
\end{lem}

\begin{proof}
    Let us first demonstrate that $\fs{Q}(x_{0,1})-(u,x,y)$ is included in the set defined by \eqref{eq:Qtildeexpr}. Pick arbitrary $x_{0,1},x_{0,2}\in\mathcal X_c$, $(u,x,y)\in\fs{Q}(x_{0,2})$ and $(\tilde u,\tilde x,\tilde y)\in\fs{Q}(x_{0,1})-(u,x,y)$, i.e. $(\tilde u,\tilde x,\tilde y)=(\bar u,\bar x,\bar y)-(u,x,y)$ for some $(\bar u,\bar x,\bar y)\in\fs{Q}(x_{0,1})$. Observe that
    \begin{align*}
        \H_x[x_{0,1}-x_{0,2};\tilde{u}] & =\H_x[x_{0,1}-x_{0,2};\bar u-u]                                                                    \\
                                        & \ifext{\color{blue}\overset{\eqref{eq:lin1}}{=}}\fi \H_x[x_{0,1};\bar u]+\H_x[-x_{0,2};-u] \\
                                        & \ifext{\color{blue}\overset{\eqref{eq:lin2}}{=}}\fi \H_x[x_{0,1};\bar u] - \H_x[x_{0,2};u] \\
                                        & = \bar x - x = \tilde x
    \end{align*}
    holds and, similarly, $\H[x_{0,1}-x_{0,2};\tilde{u}]$ equals $\tilde y$. Then, note that $\fs{Q}(x_{0,1})-(u,x,y)\subseteq \fs{U}\times\fs{X}\times\fs{Y}-(u,x,y)$ follows from \eqref{eq:Qdef}.

    Conversely, consider arbitrary $x_{0,1},x_{0,2}\in\es{X}_c$ and triple $(u,x,y)\in\fs{Q}(x_{0,2})$ and pick any $(\tilde u,\tilde x,\tilde y)$ in the set defined in \eqref{eq:Qtildeexpr}. Then, let us prove that $(\tilde u,\tilde x,\tilde y)$ also belongs to $\fs{Q}(x_{0,1})-(u,x,y)$. Define $(\bar u,\bar x,\bar y)\coloneqq (\tilde u,\tilde x,\tilde y)+(u,x,y)$ so that $(\bar u,\bar x,\bar y)\in\fs{U}\times\fs{X}\times\fs{Y}$ holds. Then, observe that $\H_x[x_{0,1};\bar u]=\H_x[x_{0,2};u]+\H_x[x_{0,1}-x_{0,2};\tilde{u}]=x+\tilde x=\bar x$ and, in an analogous way, $\H[x_{0,1};\bar u]=\bar y$. This proves that $(\bar u,\bar x,\bar y)\in\fs{Q}(x_{0,1})$ so that $(\tilde u,\tilde x,\tilde y)\in \fs{Q}(x_{0,1})-(u,x,y)$.
    %
\end{proof}

Particularizing this result for \mbox{$x_0=x_{0,1}=x_{0,2}$} allows to conclude that $\fs{W}(x_{0})-(u,y)$ reads
\begin{multline*}
    \{(\tilde{u},\tilde{y})\in\fs{U}\times\fs{Y}-(u,y):\\
    \H_x[\0;\tilde{u}]\in\fs{X}-\H_x[x_0;{u}],\;\H[\0;\tilde{u}]=\tilde y\},
\end{multline*}
for any $(u,y)\in\fs{W}(x_{0})$. This leads to a more explicit expression of $\widetilde{\fs{U}}(x_0,u)$ than \eqref{eq:Ddef}.

\begin{cor}\label{cor:Wtilde}
    Given any $x_0\in\es{X}_c$ and $u\in\fs{U}_a(x_0)$. Set $\widetilde{\fs{U}}(x_0,u)$ reads
    \begin{equation*}
        \{\tilde{u}\in\fs{U}-u:
        \H_x[\0;\tilde{u}]\in\fs{X}-\H_x[x_0;{u}],\; \H[\0;\tilde{u}]=\N\}.
    \end{equation*}
\end{cor}

Let us summarize results of this section. Given $x_0\in\es{X}_c$ and $(u,x,y)\in\fs{Q}(x_0)$, \emph{incremental} input $u+\tilde u$ is admissible for $x_0$ and leads to the same output $y$ iff $\tilde u\in\widetilde{\fs{U}}(x_0,u)$. From Cor.~\ref{cor:Wtilde}, this means that (i) $\tilde u\in\fs{U}_\Sigma$, (ii) both $\tilde u(t)\in\es{U}-u(t)$ and $\tilde x(t)\coloneqq\H_x[\0;\tilde u](t)\in\es{X}-x(t)$ hold for all $t\in\real_{\geq 0}$ and (iii) $\tilde y\coloneqq\H[\0;\tilde u]=\N$. An essential point here is that $\tilde x$ and $\tilde y$ derive from $\tilde{u}$ and \emph{zero initial condition}, i.e. $(\tilde{u},\tilde{x},\tilde{y})\in\fs{Q}_\Sigma(\0)$.

\ifext{\color{blue}
    \begin{rem}
        Just like $v_1,v_2\in\es{A}$ does not imply $v_1-v_2\in\es{A}$ for arbitrary (non linear) set $\es{A}$, let us stress that $\tilde u\in\widetilde{\fs{U}}(x_0,u)$ does not imply that $\tilde u$ is an admissible input for $x_0$ in general. 
    \end{rem}
}\fi

\subsection{Incremental characterization of the proposed taxonomy}

By making use of set $\widetilde{\fs{U}}(x_0,u)$, next lemma offers a new characterization of the taxonomy.

\begin{lem}\label{lem:kindchar}
    Pair $(x_0,y)\in \fs{S}$ is \sing{} of
    \begin{itemize}
        \item the \emph{1st kind} iff it holds
              \begin{multline} \label{eq:kind1car}
                  \forall u\in\mathfrak{H}^{-1}[x_0;y],\; \forall \tilde u\in \widetilde{\fs{U}}(x_0,u) \setminus \{\N \} : \\
                  \H_x[\0;\tilde u] = \N
              \end{multline}
        \item the \emph{2nd kind} iff it holds
              \begin{multline} \label{eq:kind2car}
                  \forall u\in\mathfrak{H}^{-1}[x_0;y],\; \forall \tilde u\in \widetilde{\fs{U}}(x_0,u) \setminus \{\N \}: \\
                  \H_x[\0;\tilde u] \neq \N
              \end{multline}
        \item the \emph{3rd kind} iff neither \eqref{eq:kind1car} nor \eqref{eq:kind2car} is valid.
    \end{itemize}
\end{lem}

\begin{proof}
    By definition, $(x_0,y)$ is \sing{} of the 1st kind iff $u\neq u_2$ and $y = y_2$ imply $x= x_2$ for all $(u,x,y),(u_2,x_2,y_2)\in\fs{Q}(x_0)$. Defining $(\tilde u,\tilde x,\tilde y)\coloneqq (u_2,x_2,y_2)-(u,x,y)$, this is equivalent to saying that $\tilde u \neq \0$ and $\tilde y = \0$ imply $\tilde x= \0$ for all $(u,x,y)\in\fs{Q}(x_0)$ and for all $(\tilde u,\tilde x,\tilde y)\in\fs{Q}(x_0)-(u,x,y)$. Dropping state-trajectories and observing that $(\tilde u,\tilde x,\tilde y)$ are associated with zero initial state (see Lem.~\ref{lem:Qtildeexpr}), this is equivalent to next condition
    \begin{multline} \label{eq:coq}
        \forall (u,y)\in\fs{W}(x_0),\; \forall (\tilde u,\tilde y)\in\fs{W}(x_0)-(u,y):\\
        \left. \begin{array}{r} \tilde u \neq \N \\ \tilde y = \N \end{array} \right\}
        \Rightarrow \H_x[\0;\tilde u] = \N.
    \end{multline}
    Eq. \eqref{eq:Ddef} and the fact that $(u,y)\in\fs{W}(x_0)$ is equivalent to $u\in\mathfrak{H}^{-1}[x_0;y]$ prove that \eqref{eq:coq} is equivalent to \eqref{eq:kind1car}. The fact that \eqref{eq:kind2car} characterizes \singy{} of the 2nd kind can be proved in a similar way. \Singy{} of the third refers to \singy{} which is neither of the first nor of the 2nd kind.
\end{proof}

\subsection{Proof of Th.~\ref{th:singNL}\label{ap:singNL}}

\ifext{
    \begin{figure}[t]
        \centering
        \includegraphics[width=.8\columnwidth]{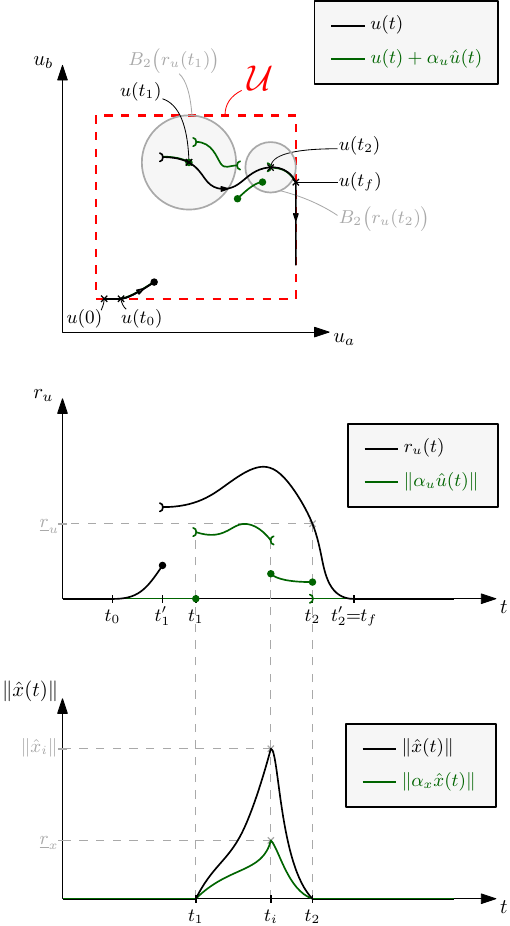}
        \caption{\color{blue}Graphical illustration of the proof of Th.~\ref{th:singNL} for $m=2$ and $\rho=0$. Components of $u$ are denoted by $u_a$ and $u_b$.} \label{fig:singNL}
    \end{figure}
}\fi

Given $x_0\in\es{X}_c$ and $(u,x,y)\in\fs{Q} (x_0)$.
By Prop.~\ref{prop:singD}, $(x_0,y)$ is \sing{} if $\card{\widetilde{\fs{U}}(x_0,u)} \geq 2$. From Cor.~\ref{cor:Wtilde}, recall that $\widetilde{\fs{U}}(x_0,u)$ gathers signals $\tilde u\in\fs{U}-u$ satisfying $\H[\0;\tilde u]=\0$ and $\H_x[\0;\tilde u]\in\fs{X}-x$. Due to the fact that $\0\in\widetilde{\fs{U}}(x_0,u)$ trivially holds, let us prove \singy{} by constructing any non zero measure signal in $\widetilde{\fs{U}}(x_0,u)$. 

\ifext{\color{blue}
    Fig.~\ref{fig:singNL} is a sketch the reader can refer to in order to ease figuring out forthcoming technical developments.
}\fi

Define closed hyper-ball $B_k(r)$ as follows:
\begin{equation*}
    B_k(r)\coloneqq \{\eta\in\real^k: \|\eta \|\leq r\}.
\end{equation*}
Eq. \eqref{eq:fenetreU} (resp. \eqref{eq:fenetreX}) is equivalent to the existence of  $r_u (t)>0$ (resp. $r_x(t)>0$) such that $u(t)+B_m(r_u (t))\subseteq \mathcal U$ (resp. $x(t)+B_n(r_x (t))\subseteq \mathcal X$). As a preliminary step, let us prove that $t\mapsto r_u(t)$ (resp. $t\mapsto r_x(t)$) admits a strictly positive lower bound $\underline r_u$ (resp. $\underline r_x$) on some interval $[t_1,t_2]\subseteq]t_0,t_f[$ satisfying $t_1<t_2$. First note that $r_u$ inherits piecewise continuity from $u$. Thus, there exists $]t_1',t_2'[\subseteq]t_0,t_f[$, with $t_1'<t_2'$ on which $r_u$ is continuous. Therefore, on any closed interval $[t_1,t_2]\subseteq ]t_1',t_2'[$ with $t_1<t_2$, function $r_u$ admits a minimum $\underline r_u$ (necessarily strictly positive) by virtue of extreme value theorem. Since $x$ is continuous on $]t_0,t_f[$, the same argument can be invoked to prove that $\underline r_x\coloneqq \textstyle\min_{t\in[t_1,t_2]}r_x (t)$ exists and is strictly positive.

Unconstrained system $\Sigma$ is IR by assumption. This implies that there exists a non zero measure input $\hat u\in\fs{U}_\Sigma$ (i) producing output $\0$ from initial state $\0$ when exerted to $\Sigma$ (i.e. $(\hat u,\hat x, \0)\in\fs{Q}_\Sigma(\0)$ for some $\hat x$) and such that (ii) $\hat u(t)=\0$ holds for all $t\in\real_{\geq 0}\setminus [t_1,t_2]$. Indeed, \cite[\characIR]{Kreiss2020} applies for $\Sigma$ and proves that $\rho=\dim{\ker{\smallmat{B \\ D}}}$ or $\dim{\es{R}(\Sigma)}$ is strictly positive. If $\rho>0$, then $\hat u$ can be selected as a continuous signal which is not identically zero and satisfies $\hat u(t)\in\ker{\smallmat{B \\ D}}$ if $t\in[t_1,t_2]$. Clearly, this signal $\hat  u$ belongs to $\fs{U}_\Sigma$ and corresponding state trajectory $\hat x$ equals $\0$. If $\rho=0$, so that $\dim{\es{R}(\Sigma)}>0$ holds, pick arbitrary $\hat x_i\in\es{R} (\Sigma)\setminus\{\0\}$ and any $t_i$ such that $t_1<t_i<t_2$. Then, there exists a piecewise continuous input transferring state from $\0$ to $\hat x_i$ and then back to the origin while maintaining output identically null \ifext{\color{blue}(see Subsection~\ref{sec:contr})}\fi. Since this motion can be performed arbitrarily fast, $\hat u$ can be selected in such a way that $\hat x(t_1)=\hat x(t_2)=\0$ and $\hat x(t_i)=\hat x_i$ hold. Note that $\hat x(t)=\0$ holds for all $t\in\real_{\geq 0}\setminus[t_1,t_2]$ because $\hat u(t)$ equals zero on the same interval. Also remark that $\hat x_i\neq \0$ implies that input $\hat u$ is not identically zero. Once again, $\hat u\in\fs{U}_\Sigma$ holds in this case.

Let us show that there exists $\alpha>0$ such that $\alpha\hat u\eqqcolon \tilde u\in\widetilde{\fs{U}}(x_0,u)$, which proves \singy{} since $\hat u$ has non zero measure. Define $\alpha_u\coloneqq  \underline r_u / b_u^+$ where $b_u^+$ is an upper bound of $t\mapsto\|\hat u(t)\|$ on compact interval $[t_1,t_2]$, which necessarily exists since $\hat u$ is piecewise continuous. Then, it holds $\alpha_u \leq r_u(t) /\|\hat u (t)\|$ and, in turn, $\alpha_u\hat u(t) \in B_m(r_u (t))$ for all $t\in[t_1,t_2]$. Bearing in mind that $B_m(r_u (t))\subseteq \mathcal U-u(t)$ is valid for all $t\in[t_1,t_2]$ and that $\hat u(t)=\0\in\es{U}-u(t)$ holds for all $t\in\real_{\geq 0}\setminus[t_1,t_2]$, this proves that $\alpha \hat u \in \fs{U}-u$  holds for all $0<\alpha\leq \alpha_u$. If $\rho=0$, then $\hat x=\0$ so that $\alpha\hat x\in\fs{X}-x$ trivially holds for all $0<\alpha\leq \alpha_u$. If $\rho>0$, continuity of $\hat x$ is ensured for any $\hat u\in\fs{U}_\Sigma$ which proves existence of strictly positive $\alpha_x\coloneqq  \underline r_x / \textstyle\max_{t\in[t_1,t_2]} \|\hat x(t) \|$. Then, similar arguments as before ensure that $\alpha_x \hat x \in \fs{X}-x$ holds for all  $0<\alpha\leq \alpha_x$. If $\rho=0$ (resp. $\rho>0$), define $\alpha\coloneqq \alpha_u>0$ (resp. $\alpha\coloneqq \min\{\alpha_u,\alpha_x\}>0$). In both cases ($\rho=0$ or $\rho>0$), observe that (i) $\H_x[\0;\alpha\tilde u]=\alpha \hat x$ holds\ifext{\color{blue} by virtue of \eqref{eq:lin2}}\fi, (ii) $\H[\0;\alpha\hat u]=\0$ is valid for similar reasons and (iii) $(\alpha\hat u,\alpha\hat x)\in(\fs{U}-u)\times(\fs{X}-x)$ holds. We have just proved that non zero measure input $\tilde u=\alpha\hat u$ belongs to $\widetilde{\fs{U}}(x_0,u)$, so that $(x_0,y)$ is \sing{}.

\subsection{Proof of Th.~\ref{th:locglob}\label{ap:locglob}}

Let us consider that {Hyp.}~\ref{hyp:linEns} is valid throughout this subsection. Th.~\ref{th:locglob} is proved by specializing previous results in this context.

Clearly, function spaces $\fs{U}$ and $\fs{X}$ inherit linearity property from $\mathcal U$ and $\mathcal X$. 
An essential consequence of this feature is given by the following lemma.
\begin{lem} \label{lem:locglob}
    Assume that {Hyp.}~\ref{hyp:linEns} holds and consider any $x_{0,1},x_{0,2}\in \mathcal X_c$. For any $(u,y)\in\fs{W}(x_{0,2})$, it holds
    \begin{equation}
        \fs{W} (x_{0,1})-(u,y) =\fs{W}(x_{0,1}-x_{0,2}) 
    \end{equation}
\end{lem}

\begin{proof}
    Let $x$ be such that $(u,x,y)\in\fs{Q}(x_{0,2})$. Since $\fs{Q}(x_{0,2})\subset\fs{U}\times\fs{X}\times\fs{Y}$, it holds  $\fs{U}\times\fs{X}\times\fs{Y}-(u,x,y)=\fs{U}\times\fs{X}\times\fs{Y}$ under {Hyp.}~\ref{hyp:linEns}. Hence, $\fs{W} (x_{0,1})-(u,y)$ reads $\{(\tilde{u}, \tilde{y})\in\fs{U}\times\fs{Y}:\H_x[x_{0,1}-x_{0,2};\tilde{u}]\in\fs{X},\;\H[x_{0,1}-x_{0,2};\tilde u]=\tilde y\}$ (see \eqref{eq:Qtildeexpr}) which is nothing but $\fs{W}(x_{0,1}-x_{0,2})$.
\end{proof}

As a result, $\widetilde{\fs{U}}(x_0,u)$ reduces to
\begin{equation} \label{eq:Utildelin}
    \widetilde{\fs{U}}\coloneqq \{\tilde u:\real_{\geq 0}\rightarrow \real^m \mid (\tilde u,\N)\in \fs{W} (\0) \}
\end{equation}
whenever {Hyp.}~\ref{hyp:linEns} is valid (see \eqref{eq:Ddef}). This set gathers inputs which produce an identically zero output from zero initial state and comply with the constraints. As suggested by the notation, note that $\widetilde{\fs{U}}$ depends neither on $x_0$ nor on $u$ anymore.

Let us particularized Prop.~\ref{prop:singD} and Lem.~\ref{lem:kindchar} in this context.

\begin{lem} \label{lem:kindcarlin}
    Assume that {Hyp.}~\ref{hyp:linEns} holds. Pair $(x_0,y)\in\fs{A}$ is \sing{} iff it holds:
    \begin{equation} \label{eq:carlin}
        \card{\widetilde{\fs{U}} \setminus \{ \N \}}\geq 1.
    \end{equation}
    Furthermore, $(x_0,y)$ is of:
    \begin{itemize}
        \item the \emph{1st kind} iff it holds
              \begin{equation} \label{eq:kind1carlin}
                  \forall \tilde u\in \widetilde{\fs{U}}\setminus \{\N \}: \H_x[\0;\tilde u] = \N
              \end{equation}
        \item the \emph{2nd kind} iff it holds
              \begin{equation} \label{eq:kind2carlin}
                  \forall \tilde u\in \widetilde{\fs{U}} \setminus \{\N \}: \H_x[\0;\tilde u] \neq \N
              \end{equation}
        \item the \emph{3rd kind} iff neither \eqref{eq:kind1carlin} nor \eqref{eq:kind2carlin} is valid.
    \end{itemize}
\end{lem}

Since \eqref{eq:carlin}, \eqref{eq:kind1carlin} and \eqref{eq:kind2carlin} depend neither on $x_0$ nor on $y$, previous properties are independent of the choice of $(x_0,y)$, provided that $y$ is admissible for $x_0$ so that $\mathfrak{H}^{-1}[x_0;y]$ is non empty. Statement of Th.~\ref{th:locglob} follows immediately.

\ifext{\color{blue}
    \begin{rem}[$\fs{W} (x_0)$ and $\mathfrak{H}^{-1}{[}x_0;y{]}$ are affine]
        Under {Hyp.}~\ref{hyp:linEns}, it can easily be checked that $\fs{Q} (\0)$ and, in turn, $\fs{W} (\0)$ are linear spaces. In view of Lem.~\ref{lem:locglob}, this means that $\fs{Q} (x_0)$ and, in turn, $\fs{W} (x_0)$ are affine spaces, for all $x_0\in\es{X}$. Indeed, set $\fs{Q} (x_0)$ can be written as $\fs{Q} (\0)+(u,x,y)$ for any $(u,x,y)\in\fs{Q} (x_0)$. It also follows that $\widetilde{\fs{U}}$ is linear, so that $\mathfrak{H}^{-1}[x_0;y]$ is affine since this last set equals $\widetilde{\fs{U}} + u$ for all $u\in\mathfrak{H}^{-1}[x_0;y]$ (see \eqref{eq:eqCDinc}).
    \end{rem}

    \begin{rem}
        Note that Lem.~\ref{lem:Ebij} allows to construct an \emph{admissible} input $\tilde u$ with $\mathfrak S$ enjoying properties analoguous to that of $w$: Denoting $E(F,\0;(w,\eta,\0))$ by $(\tilde u,\tilde x,\tilde y)$, it follows that $\tilde u=RLw+RFT\eta$ is admissible and produces $\tilde y=\0$ for $x_0=\0$, i.e. $\tilde u\in\widetilde{\fs{U}}$ (see \eqref{eq:Utildelin}). Existence of input $\tilde u$ has crucial consequence since it proves IR. 
    \end{rem}

    \begin{rem}
        IR of $\mathfrak S$ can be alternatively qualified by means of sets $\fs{G} \coloneqq \fs{Q} (\0) \cap (\fs{U}\times \fs{X}\times \{\N\} )$ and $\fs{G}_0 \coloneqq \fs{Q} (\0) \cap (\fs{U} \times \{\N\} \times \{\N\} )$. Indeed, Lem.~\ref{lem:kindcarlin} proves that $\mathfrak S$ is IR (i) of the 1st kind if $\fs{G}  \subseteq \fs{G}_0$, (ii) of the 2nd kind if $\fs{G} \cap \fs{G}_0 = \{\N\}$ and (iii) of the 3rd kind if none of those two relationships is valid.
    \end{rem}
}\fi

\section{Proof of Cor.~\ref{cor:notFenetre}\label{ap:notFenetre}}

As a preliminary step, consider following lemma.
\begin{lem}\label{lem:Czint}
    Given open set $\es{Z}$ and signal $z:\real_{\geq 0}\rightarrow \es{Z}$. Let $\es{C}_z\neq\varnothing$ be the largest open subset of $\real_{\geq 0}$ where $z$ is continuous. Then, $\mathcal Q_z\coloneqq\{t\in\es{C}_z:z(t)\in\es{Z}\}$ is open.
\end{lem}

\begin{proof}
    Pick any $t_0\in\mathcal Q_z$. Since $\es{Z}$ is open, there exists an open hyper-ball $B$ centered on $z(t_0)$ such that $B\subseteq \es{Z}$. From $t_0\in\es{C}_z$, there exists a connected open set $\mathcal T\subseteq \es{C}_z$ on which $z$ is continuous and satisfying $t_0\in\es{T}$. Observe that if $z(\mathcal T)\subset B$, then $\mathcal T\subseteq \mathcal Q_z$. Otherwise, continuity of $z$ on $\mathcal T$ ensures existence of open set $\hat{\mathcal T}\subseteq\mathcal T$ satisfying $t_0\in\hat{\es{T}}$ and such that $z(\hat{\mathcal T})\subset B$ so that $\hat{\mathcal T}\subseteq \mathcal Q_z$.
\end{proof}

Define the following sets:
\begin{align*}
    \es{T}_u & \coloneqq \{t\in\real_{\geq 0}:u(t)\in\es{U}\setminus\text{int}(\es{U})\}, \\
    \es{T}_x & \coloneqq \{t\in\real_{\geq 0}:x(t)\in\es{X}\setminus\text{int}(\es{X})\}.
\end{align*}
From discussion above Cor.~\ref{cor:notFenetre}, $\es{T}_u$ (resp. $\es{T}_u\cup\es{T}_x$) is dense in $\real_{\geq 0}$ if $\rho>0$ (resp. $\rho=0$).

Assume that $\rho>0$ holds. Then, $\es{T}_u\cap \es{C}$ is also dense in $\es{C}$, so that closure of $\es{T}_u\cap \es{C}$, denoted by $\overline{\es{T}_u\cap \es{C}}$, equals $\es{C}$ holds.
\ifext{\color{blue}
    Indeed, for all open set $\es{O}\subseteq \real_{\geq 0}$, it holds $\es{O}\cap \es{T}_u \neq \varnothing$, so that for all open set $\es{O}\subseteq \es{C}\subseteq\real_{\geq 0}$, set $\es{O}\cap \es{T}_u=(\es{O}\cap \es{T}_u) \cap \es{C}=\es{O}\cap (\es{T}_u \cap \es{C})$ is distinct from $\varnothing$.
}\fi
Furthermore, observe that $\es{C}\setminus \es{T}_u$ equals $\{t\in\es{C}:u(t)\in\text{int}(\es{U})\}$ 
which is open, by virtue of Lem.~\ref{lem:Czint}. This, in turn, proves that $\es{T}_u\cap \es{C}$ is a closed subset \emph{of} $\es{C}$. As a result, $\es{T}_u\cap \es{C}=\overline{\es{T}_u\cap \es{C}}=\es{C}$, so that \eqref{eq:notfenetreU} holds for all $t$ where $u$ is continuous.

The case $\rho=0$ can be handled similarly. It holds $\overline{(\es{T}_u\cup \es{T}_x)\cap \es{C}}=\es{C}$. Set $\es{C}\setminus (\es{T}_u\cup \es{T}_x)$ equals
\begin{multline*}
    \{t\in\es{C}:u(t)\in\text{int}(\es{U}) \wedge x(t)\in\text{int}(\es{X}) \}= \\
    \{t\in\es{C}:u(t)\in\text{int}(\es{U})\} \cap \{t\in\es{C}:x(t)\in\text{int}(\es{X})\}
\end{multline*}
which can be proved to be open by Lem.~\ref{lem:Czint} since $x$ in continuous on $\es{C}$. This yields $(\es{T}_u\cup \es{T}_x)\cap \es{C}=\overline{(\es{T}_u\cup \es{T}_x)\cap \es{C}}=\es{C}$, so that either \eqref{eq:notfenetreU} or \eqref{eq:notfenetreX} holds for all $t\in\es C$.

\end{document}